\documentclass[11pt,onecolumn]{article}
\renewcommand{\baselinestretch}{1.35}
\usepackage[utf8]{inputenc}

\usepackage{times}
\usepackage{graphicx,subfigure}
\usepackage{amsmath,amssymb}
\usepackage{cite}
\usepackage{hyperref}
\usepackage{url}
\usepackage{dsfont}

\usepackage{bbm}
\usepackage{color}
\usepackage{algorithm}
\usepackage{algpseudocode}
\usepackage{mathtools}
\DeclarePairedDelimiter{\abs}{\lvert}{\rvert}
\usepackage[font=small,justification=centering]{caption}
\ifpdf
    \graphicspath{{figures/PNG/}{figures/PDF/}{figures/}}
\else
    \graphicspath{{figures/EPS/}{figures/}}
\fi

\newtheorem{theorem}{Theorem}
\newtheorem{lemma}{Lemma}
\newtheorem{corollary}{Corollary}
\newtheorem{definition}{Definition}
\newtheorem{remark}{Remark}

\def \Ju {{\J_l^j(\delta_l^j,\bar\delta_l^{\sf r},U_l^j) }}

\def \bU {{\bf{U}}}
\def \X {{\bf{ X}}}
\def \Y {{\bf{{Y}}}}

\def \bdelta {{\boldsymbol{\delta}}}

\def \cC {{\cal C}}
\def \cF {{\cal F}}
\def \cX {{\cal X}}
\def \cY {{\cal Y}}
\def \cS {{\cal S}}
\def \J {{\cal J}}
\def \N {{\cal N}}

\def \U {{\cal U}}
\def \Pt {{\cal P}}
\def \Q {{\cal Q}}
\def \R {{\cal R}}
\def \Jc {{\cal J}}

\def \H {{\sf H}}
\def \h {{\sf h}}
\def \Pc {{\sf P}}
\def \T {{\sf T}}
\def \D {{\sf D}}
\def \d {{\sf d}}

\def \Ec {{\sf E}}
\def \C {{\sf C}}

\def \P {{\mathbb {P}}}
\def \J {{ J}}
\def \d{{\rm d}}
\def \E {{\mathbb {E}}}
\def \LR {{\sf LR}}
\def \cpi {{\cal \pi}}
\def \med{\;|\;}
\def\tenpt{   \def\baselinestretch{.9}\let\normalsize\normalsize\normalsize}

\DeclareMathOperator*{\argmax}{arg\,max}

\DeclareMathOperator*{\argmin}{argmin}

\title{\vspace{-.4 in}\bf \LARGE Secure Estimation under Causative Attacks}
\date{}
\author{Saurabh Sihag \and Ali Tajer \thanks{Electrical, Computer, and Systems Engineering Department, Rensselaer Polytechnic Institute, Troy, NY 12180.}}

\begin{document}

\maketitle
\allowdisplaybreaks

\begin{abstract}
This paper considers the problem of secure parameter estimation when the estimation algorithm is prone to \textbf{\textsl{causative}} attacks.  Causative attacks, in principle, target decision-making algorithms (e.g., inference and learning algorithms) to alter their decisions by making them oblivious to specific attacks. Such attacks influence inference algorithms by tampering with the mechanism through which the algorithm is provided with the statistical model of the population about which an inferential decision is made. Causative attacks are viable, for instance, by contaminating the historical or training data, or by compromising an expert who provides the model. In the presence of causative attacks, the inference algorithms operate under a distorted statistical model for the population from which they collect data samples.  This paper introduces specific notions of secure estimation and provides a framework under which secure estimation under causative attacks can be formulated. A central premise underlying the secure estimation framework is that forming secure estimates introduces a new dimension to the estimation objective, which pertains to detecting attacks and isolating the true model. Since detection and isolation decisions themselves are imperfect, their inclusion induces an inherent coupling between the desired secure estimation objective and the auxiliary detection and isolation decisions that need to be formed in conjunction with the estimates. This paper establishes the fundamental interplay among the decisions involved, and characterizes the general decision rules in closed-forms for any desired estimation cost function. Furthermore, to circumvent the computational complexity associated with growing parameter dimension or attack complexity, a scalable estimation algorithm and its attendant optimality guarantees are provided. Finally, the theory developed is applied to secure parameter estimation in a sensor network operating in an adversarial environment.

\end{abstract}

\section{Introduction}
\label{sec:introduction}

\subsection{Motivation}
\label{sec:overview}

Statistical inference offers mechanisms for deducing the statistical properties of a population based on the data sampled from the population. Inference problems, broadly, focus on discerning the statistical model of the population or forming estimates about an unknown parameter that specifies the statistical model of the population. Anomaly detection,  which has diverse applications in intrusion detection, fraud detection, fault detection, system health monitoring, and event detection, constitutes a major class of inference problems in which the objective is raising alarms when the data pattern (e.g., statistical model) deviates significantly from the expected patterns. Effective detection of anomalies in the data strongly hinges on the known rules for distinguishing normal and abnormal data segments. These rules, for instance, can be specified by an expert or by leveraging the historical data, depending on the context of the application. These rules, subsequently, can be used for designing inference algorithms for detecting the anomalies with specified guarantees on relevant figures of merit.


While anomaly detection, which in essence copes with the vulnerability of the sampled data to being contaminated or compromised, has been studied extensively (c.f. \cite{J21,vempaty,vempatytarget}), the vulnerability of the {\sl inference algorithms} to being compromised is far less-investigated. The nature of security vulnerabilities that inference algorithms are exposed to is fundamentally distinct from that of data. Specifically, in the case of compromised sample data, the information of the decision algorithm about the model remains intact, while the data fed to the algorithm is anomalous, in which case the counter-measures consist of winnowing out the compromised samples and forming inferential rules that are robust against them. In contrast, attacks on the algorithms can be exerted by providing the algorithm with an incorrect statistical model for the data. This is viable by, for instance, contaminating the historical data or by confusing the expert that produces a model, which are critical for furnishing the true model for the statistical model of the data. Therefore, when the sampled data is compromised, an inference algorithm produces decisions based on an un-compromised known model for the data, while the data that it receives and processes are compromised. On the other hand, when the historical data or the expert are compromised, an inference algorithm functions based on an incorrect model for the data, in which case even un-compromised sampled data produces unreliable decisions.

The aforementioned security vulnerabilities for the inference algorithms can be capitalized on by adversaries in order to force an inference algorithm to deviate from its optimal structure and produce decisions in ways that serve an adversary's purposes. Such attacks on decision algorithms are often referred to as {\bf causative attacks}, through which an adversary aims to (i) make the inference algorithms oblivious to specific attacks, or (ii) degrade the performance of the inference algorithm in the presence of such an attack~\cite{Barreno}. 

While the notion of secure decision-making in adjacent domains (e.g., machine learning and data mining) has been heavily investigated in recent years, the fundamental limits of secure statistical inference are not well-investigated, and all the limited existing studies remain rather ad-hoc. In this paper, we provide a framework for secure parameter estimation under the potential presence of {\sl causative} attacks. We establish the fundamental tradeoffs involved in decision-making under causative attacks and characterize the optimal decision rules for securely estimating the parameters and concurrently detecting the presence of the attackers. Furthermore, we provide scalable algorithms for addressing settings in which the size or complexity of the attacks grow, and provide optimality guarantees on the performance of these algorithms. A summary of the content and the contributions are discussed in Section~\ref{sec:contributions}, and the relevant literature on secure statistical inference are reviewed in Section~\ref{sec:review}.

\subsection{Overview and Contributions}
\label{sec:contributions}

To lay the context for discussing the problem investigated, consider the canonical parameter estimation problem in which we have a collection of probability distributions $\{P_X: X\in\cX\}$ defined over a common measurable space. The objective is to estimate $X$, which lies in a known set $\cX\subseteq\mathbb{R}^p$, from data samples $\Y\triangleq [Y_1,\dots, Y_n]$, where the sample $Y_r$ is distributed according to $P_X$ and  lies in a known set $\cY\subseteq\mathbb{R}^m$. We denote the probability density functions (pdfs) that the statistician {\sl assumes} about the underlying distributions of $X$ and $Y_r$ by $\pi$ and $f (\cdot \med X)$, respectively, i.e., 
\begin{align}
Y_r\;\sim\; f (\cdot \med  X)\ ,\qquad \mbox{with} \quad X\;\sim\; \pi\ .
\end{align}
For convenience in the analysis, we will assume that the pdfs do not have any non-zero probability masses over lower-dimensional manifolds. The objective of the statistician is formalizing a reliable estimator 
\begin{align}
\label{eq:estimation}
\hat X(\Y):\cY^{n} \mapsto \cX\ .
\end{align}

\noindent \textbf{\textsl{Causative Attacks:}}  In an adversarial environment, a malicious attacker might launch a \textsl{\textbf{causative}} attack to influence (degrade) the quality of $\hat X(\Y)$. The purpose of such an attack is to compromise the {\sl process that underlies acquiring the statistical models}. We emphasize that such an attack is different from those that aim to compromise the {\sl data}, e.g., false data injection attacks, which aim to distort the data samples $\Y$. Consequently, the effect of a causative attack is misleading the statistician about the true model $f (\cdot \med  X)$ that it assumes about the data. Such attacks are possible by compromising the historical (or training) data that is used for defining a model for the data. Depending on the specificity and the extent of a causative attack, e.g., the fraction of the historical or training data that is compromised, the true model $f (\cdot \med X)$ can deviate to alternative forms, the space of which we denote by $\cF$. The attack can affect the statistical distribution of any number of the $m$ coordinates of $\Y$. There are two major aspects to selecting $\cF$ as viable model space. 
\begin{itemize}
\item An attack is effective if the compromised model is sufficiently distinct from the model assumed by the statistician. Hence, even though in general $\cF$ can be any representation of possible kernels $f (\cdot \med X)$ mapping $\cY$ to $\mathbb{R}^m$, only a subset of such  mappings suffices to describe the set of effective attacks.
\item There exists a tradeoff between the complexity of the model space and its expressiveness. Specifically, if it is overly expressive, it can represent the possible compromised models with a more refined accuracy at the expense of more complex inferential rules. 
\end{itemize}
 The specifics of the attack model will be discussed in Section~\ref{sec:model}. 
\vspace{.1 in}
\newpage

\noindent \textbf{\textsl{$(q,\beta)$-Security:}} The potential presence of an adversary introduces a new dimension to the estimation problem in~\eqref{eq:estimation}. Specifically, on the one hand, the stochastic model of the data can be altered by an attack and detecting whether the data model is compromised, itself being an inference task,  is never perfect. On the other hand,  designing an optimal estimation rule strongly hinges on successfully isolating the true model. Hence, there exists an inherent coupling between the original estimation  problem of interest and the introduced auxiliary problem (i.e., detecting the presence of an attacker and isolating the true model). Based on this observation, in an adversarial setting, there exists uncertainty about the true model, based on which the quality of the estimator is expected to degrade with respect to an attack-free setting. We are interested in establishing the fundamental interplay between the quality of discerning the true model and the degradation level in the estimation quality. To establish this interplay, we say that an estimator is \textbf{\textsl{$(q,\beta)$-secure}} if its estimation cost is weaker than that of the attack-free setting by a factor $q\in[1,+\infty)$, while missing at most $\beta \in (0,1]$ fraction of the attacks\footnote{Estimation costs and the associated estimation degradation factor $q$ will be defined  in Section~\ref{sec:definitions}.}.

In this paper, we pursue three intertwined objectives. First, we characterize the fundamental tradeoffs between $q$ and $\beta$ and delineate a tradeoff curve. Secondly, we characterize the inference rules in closed-forms and provide a secure estimation algorithm that achieves the optimal tradeoffs for any desired point on the tradeoff curve. Finally, to circumvent the computational complexity as the the dimension of the data, $p$, grows, or the complexity of the attacks scales up (e.g., the number of coordinates compromised),  we provide a scalable algorithm that has low computational complexity with guaranteed optimality in the asymptote of large data dimension $p$.

\subsection{Related Studies}
\label{sec:review}
The problem of secure inference is studied primarily in the context of sensor networks. The study in \cite{wilson2016}, in particular, considers a two-sensor network in which one sensor is known to be secured, and one sensor is vulnerable to attacks. The objective is forming an estimate based on the mean-squared error criterion, for which a heuristic detection-driven estimator is designed. According to this design, first, a decision is formed about whether the unsecured sensor is attacked. If it is deemed to be attacked, then the estimator will rely only on the secured sensor, and otherwise, it uses the data from both sensors. Unlike in \cite{wilson2016}, we consider a model with arbitrary size, assume that all data coordinates are vulnerable to the attack, and characterize the optimal decision structure, which turns out to be different from being a detection-driven design studied in~\cite{wilson2016}. Through a case study, we will also show the rather significant improvement in the estimation quality when using the optimal rules, compared to the rules specified in~\cite{wilson2016}.  

The adversarial setting defined in this paper is also similar to the widely-studied Byzantine attack models in sensor networks, in which the data generated by the compromised sensors are modified arbitrarily by the adversaries in order to degrade or the inference quality. An overview of the impact of Byzantine attacks on inference quality in sensor networks and relevant mitigation strategies are discussed in~\cite{vempaty}. Detection-driven estimation strategies (i.e., when attack detection precedes the estimation routine) when the effect of the Byzantine attacks are characterized by randomly flipping of the information bits generated by the sensors are discussed in~\cite{vempatylocalization, manet, asymptote, distributedlarge}. Furthermore, attack-resilient target localization strategies are investigated in~\cite{vempatytarget} and \cite{vempatylocalization}, where it is assumed that the attacker adopts a fixed strategy for maximum disruption to the inference. In these studies, nevertheless, an attacker may deviate from the worst-case strategy of incurring the maximum damage in order to launch a less powerful but sustained attack, which may not be detected perfectly. Finally, strategies for isolating the compromised nodes in sensor networks are investigated in~\cite{rawat,decentralized,vempatyadaptive}. The emphasis of these studies is primarily focused on detecting attacks, or isolating the attacked sensors, which is different from the scope of our paper, which is focused on parameter estimation.

The problem of secure estimation in linear {\sl dynamical} systems in the context of cyber-physical systems has been studied extensively in the recent years (c.f. \cite{cps1,cps2,cps3,cps4,cps5,2015secure,cps6}). The studies that are more relevant to the scope of this paper include~\cite{cps2}, \cite{2015secure},~and~\cite{cps6}, which focus on the robust estimation of the states in dynamic systems. Specifically, a coding-theoretic interplay between the number of sensors compromised and the guarantees on perfect state recovery is characterized in~\cite{cps2}, a Kalman filter-based approach for identifying the most reliable set of sensors to make an inference from is investigated in~\cite{2015secure}, and designing estimators that are robust against dynamical model uncertainty is studied in~\cite{cps6}. The degradation in estimation performance in a dynamical system consisting of a single sensor network is studied from the adversary's perspective in \cite{kalman}, where the bounds on the degradation in estimation performance with degrees of stealthiness of the attacker are characterized.

All the aforementioned studies that involve secure estimation, irrespective of their focus or objective, conform in their design principle, which decouples the estimation decisions from all other decisions involved (e.g., attack detection or attacked sensor isolation), and produces either detection-driven estimators or estimation-driven detection routines. In the detection-driven estimation routines, an initial decision regarding the presence of an adversary (e.g., based on Neyman-Pearson theory) is followed by an optimal estimator based on the detection decision (e.g., Bayesian estimation). Such approaches implicitly assume that the detection decision has been perfect. Similarly, in an estimation-driven approach, the unknown parameter is first estimated, and then a detection decision is made (e.g., the generalized likelihood ratio test). Such approaches achieve optimality conditions only asymptotically, i.e., under the assumption of having an infinite number of measurements.  The premise that decoupling such intertwined estimation and detection problems into independent estimation and detection routines is sub-optimal is well-investigated (e.g., in~\cite{middleton,1992,2012joint,2012minimax}).

\section{Data Model and Definitions}
\label{sec:model}

\subsection{Attack Model}

Our focus is on the canonical estimation problem in~\eqref{eq:estimation}. The objective is to form an optimal estimate $\hat X(\Y)$ (under the general cost functions specified later) in the potential presence of a causative attack. Under the attack-free setting, the data is assumed to be generated according to the known distribution
\begin{align}
\label{eq:distribution}
Y_r\;\sim\; f (\cdot \med X)\ ,\qquad \mbox{with} \quad X\;\sim\; \pi\ , \quad\text{ for } r \in \{1,\dots n\}\;.
\end{align}
In an adversarial setting, an adversary, depending on its strength and preference, can launch an attack that can compromise the underlying process that the statistician uses for acquiring $f(\cdot \med X)$. An attack will be carried out for the ultimate purpose of degrading the estimation quality of $X$. We assume that the adversary can corrupt the data model of {\sl up to} $K \in\{1,\dots,m\}$ coordinates of $\Y$. Hence, for a given $K$, there exist $T = \sum_{i=1}^{K}{m \choose i}$ number of attack scenarios under which the compromised data models are distinct. Define $\cS \triangleq \left\lbrace S_1,\dots, S_T \right\rbrace$ as the set of all possible combinations of attack scenarios, where $S_i\subseteq\{1,\dots,m\}$ describes the set of coordinates the models of which are compromised under  scenario $i\in\{1,\dots,T\}$.

Under the attack scenario $i\in\{1,\dots,T\}$, the joint distribution of $Y_r$ deviates from $f$ and changes to a model in the space $\cF_i$. As discussed earlier, there exists a tradeoff between the expressiveness of this space and the complexity of the ensuing inferential rules. Specifically, a larger space $\cF_i$ can distinguish different attack strategies with a more accurate resolution at the expense of high complexity in the analysis and the resulting decision rules. Also, the model can be effective if it encompasses sufficiently distinct models. Throughout the analysis of the paper, we assume that $\cF_i\triangleq \{f_i(\cdot \med X)\}$, i.e., $\cF_i$ consists of one alternative distribution. This is primarily for the convenience in notations, and all the results presented can be generalized to any arbitrary space with countable elements. Based on this model, when the data models in the coordinates contained in $S_i$ are compromised, the joint distribution changes from  $f(\cdot \med X)$ to $f_i(\cdot \med X)$.

Different attack scenarios might occur with different likelihoods, e.g., compromising one coordinate is easier than compromising two, and it might turn out to be more likely. To distinguish such likelihoods we adopt a Bayesian framework in which we define  $\epsilon_0$ as the prior probability of having an attack-free scenario and define $\epsilon_i$ as the prior probability of the event that the attacker compromises the model under the coordinates specified by $S_i$. A block diagram of the attack model and the inferential goals to be characterized, which are discussed in the remainder of this section, is depicted in Fig.~\ref{fig:diagram}. Finally, we define the marginal pdf of the data at coordinate $l\in\{1,\dots,m\}$ under the attack-free setting and when the coordinate is compromised by $g_l^0$ and $g_l^1$, respectively.


\begin{figure}[t]
\centering
\includegraphics[width=6.5 in]{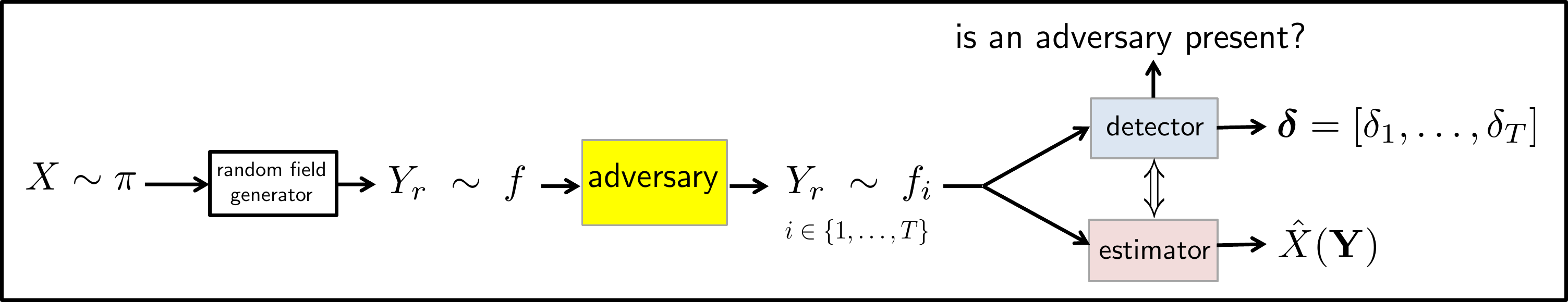}
\caption{The effect of the adversary on the data model, and the inferential decisions involved.}
\label{fig:diagram}
\end{figure}

\subsection{Compound Decisions}
The estimation objective constantly faces the uncertainty about whether an adversary exists. Furthermore, when one is deemed to exist, there is additional uncertainty pertaining to the number and the locations of the coordinates in which the data is being compromised. Hence, forming the estimate $\hat X(\Y)$ is inherently entwined with discerning the true model of the data. Decoupling the decisions for isolating the model and estimating the parameter under the isolated model does not generally render optimal performance. In fact, there exist extensive studies on formalizing and analyzing such compound decisions, which generally aim to decouple the inferential decisions. In~\cite{1992}, it is shown that generalized likelihood ratio test (GLRT) is not always optimal and necessary and sufficient conditions for its asymptotic optimality are provided. Moreover, note that GLRT utilizes maximum likelihood estimates of unknown parameters in its decision rule and thus, it is primarily focused on the detection performance. In~\cite{middleton}, the problem of signal detection and estimation in noise is investigated under the Bayes criterion.  In~\cite{2012joint} and~\cite{2012minimax}, non-asymptotic frameworks  for optimal joint detection and estimation are developed. Specifically, in~\cite{2012joint}, a binary hypothesis testing problem is investigated in which one hypothesis is composite and consists of an unknown parameter to be estimated. In~\cite{2012minimax}, the theory in~\cite{2012joint} is extended to a composite binary hypothesis testing problem in which both hypotheses correspond to composite models. 


\subsection{Decision Cost Functions}
\label{sec:definitions}

\subsubsection{Attack Detection Costs}

The possibility of having multiple alternatives to the attack-free model renders the model detection problem as the following $(T+1)-$ composite hypothesis testing problem.
\begin{equation}\label{eq:hyp1}
\begin{array}{ll}
\vspace{0.1in}
\H_0\ : & \Y \sim f(\Y \mid X),\ \text{ with }X \sim \pi(X)\\
\H_i\ : & \Y \sim f_i(\Y \mid X),\ \text{ with }X \sim \pi(X)\ , \;\quad \mbox{for}\;\;\in\{1,\dots, T\} 
\end{array} \ .
\end{equation} 
where $\H_0$ is the hypothesis corresponding to the attack-free setting, and $\H_i$ is the hypothesis corresponding to an attack launched at the coordinates in $S_i \in \cS$. For the convenience in notations, throughout the rest of the paper we denote the attack-free data model by $f_0(\cdot\med X)$, i.e., $f_0(\cdot \med X)=f(\cdot \med X)$. To formalize relevant costs for the detection decisions, we define $\D \in \left\lbrace \H_0,\dots,\H_T \right\rbrace$ as the decision on the hypothesis testing problem in (\ref{eq:hyp1}), and $\T \in \left\lbrace \H_0,\dots,\H_T \right\rbrace$ as the true hypothesis. We adopt a general {\sl randomized} test $\boldsymbol {\delta}(\Y)\triangleq [ \delta_0(\Y),\dots,\delta_T(\Y) ]$ for discerning the true hypothesis, where $\delta_i(\Y) \in [0,1]$ denotes the probability of deciding in favor of $\H_i$. Clearly,
\begin{align}
\sum\limits_{i=0}^T \delta_i(\Y) = 1\; .
\end{align}
Hence, the likelihood of deciding in favor of $\H_j$ while the true model is $\H_i$ is given by 
\begin{align}\label{e_1}
\P(\D\!=\!\H_j\!\mid \!\T\!=\!\H_i) = \displaystyle \int_{\Y}\delta_j(\Y)f_i(\Y)\; \d\Y\; .
\end{align}
We define $\Pc_{\sf md}$ as the aggregate probability of incorrectly identifying the true model under the presence of compromised coordinates, i.e.,
\begin{align}
\Pc_{\sf md} (\bdelta)&\triangleq \P(\D \neq \T\mid \T \neq \H_0)\nonumber\\
&= \frac{1}{\P(\T \neq \H_0)} \sum\limits_{i=1}^{T} \P(\D \neq \H_i \mid \T = \H_i) \P(\T=\H_i) \label{eq:md1}\\
& = \sum\limits_{i=1}^{T} \frac{\epsilon_i }{1 - \epsilon_0}\cdot  \P(\D \neq \H_i \mid \T = \H_i)\ .\label{eq:md11}
\end{align}
Furthermore, we define $\Pc_{\sf fa}$ as the aggregate probability of erroneously  declaring that a set of coordinates are compromised, while operating in an attack-free scenario. We have
\begin{align}\label{eq:fa1}
\Pc_{\sf fa} (\bdelta) &\triangleq \P(\D \neq \H_0 \mid \T = \H_0) =\sum\limits_{i=1}^T \P(\D\!=\!\H_i\!\mid \!\T\!=\!\H_0)\;.
\end{align}

\subsubsection{Secure Estimation Costs}

In this subsection, we define two estimation cost functions for capturing the  fidelity of the estimate $\hat X(\Y)$ we aim to form for $X$. For this purpose, we adopt a generic and {\sl non-negative} cost function $\C(X,U(\Y))$ to quantify the discrepancy between the ground truth $X$ and a generic estimator $U(\Y)$. Based on this, corresponding to each model $\H_i$ and given data $\Y$ we first define the {\sl average posterior cost function} as
\begin{align}\label{eq:pcost}
\C_{{\rm {p}},i}(U(\Y) \mid \Y) \triangleq \E_{i} \left[ \C(X,U(\Y)) \mid \Y \right]\; ,\qquad \forall i\in\{0,\dots,T\} \ ,
\end{align}
where the conditional expectation is with respect to $X$ when the true model is $\H_i$. Besides this, due to having distinct data models under different attack models, we consider having possibly distinct estimators under different models. Driven by this, we consider the design for an estimate for $X$ under each model. We denote the estimate of $X$ under model $\H_i$ by $\hat X_i(\Y)$, and accordingly, we define
\begin{align}\label{eq:X}
\hat \X(\Y) \triangleq [\hat{X}_0(\Y),\dots, \hat{X}_T(\Y)]\ .
\end{align}
Considering such distinct estimators, the estimation cost $\C(X,\hat X_i(\Y))$ is relevant only if the decision is $\H_i$. Hence, for any generic estimator $U_i(\Y)$ of $X$ under model $\H_i$, we define the {\sl decision-specific average cost function} as
\begin{align}\label{eqr1}
J_i(\delta_i,U_i(\Y)) \triangleq \E_i[\C(X,U_i(\Y))\mid \D = \H_i]\; ,\qquad \forall i \in\{0,\dots,T\} \ ,
\end{align}
where the conditional expectation is with respect to $X$ and $\Y$. Accordingly, we define an aggregate average estimation cost according to
\begin{align}\label{eqr2}
J(\boldsymbol{\delta},\bU) \triangleq \max_{i \in \{0,\dots,T\}} J_i(\delta_i,U_i(\Y))\;,
\end{align}
where we have defined $\bU \triangleq [U_0(\Y),\dots,U_T(\Y)]$. Finally, corresponding to the attack-free scenario, in which the only possible data model is the assumed model $f$, corresponding to any generic estimator $V(\Y)$ we define the average estimation according to
\begin{align}\label{eq:J0}
J_0(V) = \E[\C(X,V(\Y))]\; ,
\end{align}
where the expectation is with respect to $X$ and $\Y$ under model $f$. It is noteworthy that $J_0$ defined in~\eqref{eq:J0} is fundamentally different from $J(\bdelta,\bU)$ defined in~\eqref{eqr2}, since the former is the estimation cost when there is no alternative to $f$ (i.e., the attack-free scenario), while the latter is the estimation cost in an adversarial setting in which we have decided that the attacker has not compromised the data, which being a detection decision is never perfect and can be inaccurate with a non-zero probability. The role of $J_0(V)$ in our analysis is furnishing a baseline for the estimation quality in order to assess the impact of the potential presence of an adversary on the estimation quality. 

\section{Secure Parameter Estimation}

In this section, we formalize the secure estimation problem. The core premise underlying the notion of secure estimation presented is that there exists an inherent interplay between the quality of estimating $X$ and the quality of isolating the true model governing the data. Specifically, perfect detection of an adversary's attack model is impossible. At the same time, the estimation quality strongly relies on  the successful isolation of the true data model. Lack of a perfect decision about the data model is expected to degrade the estimation quality compared to the attack-free scenario. To quantify such an interplay as well as the degradation in estimation quality with respect to the attack-free scenario, we provide the following definition. 
\begin{definition}[Estimation Degradation Factor]\label{def:EDF}
For a given estimator $V$ in the attack-free scenario, and a secure estimation procedure specified by the detection and estimation rules $(\bdelta,\bU)$ in the adversarial scenario, we define the estimation degradation factor (EDF) as
\begin{align}\label{eq:q}
q(\bdelta,\bU,V) \triangleq \frac{\J(\bdelta,\bU)}{\J_0(V)}\ .
\end{align}
\end{definition}
Based on this definition, next we define the performance region, which encompasses all the pairs of decision qualities $q(\bdelta,\bU,V)$ and $\Pc_{\sf md}(\bdelta)$ over the space of all possible decision rules $(\bdelta,\bU,V)$. 
\begin{definition}[Performance Region]\label{def:PR}
We define the performance region as the region of all simultaneously achievable estimation quality $q(\bdelta,\bU,V)$ and detection performance $\Pc_{\sf md}(\bdelta)$. 
\end{definition}
By leveraging the characteristics of the performance region, next we define the notion of $(q,\beta)$-security, which is instrumental in defining the secure estimation problem of interest. For this purpose, note that EDF normalizes the estimation cost in the adversarial setting by that of the attack-free scenario. The two estimation cost functions involved in $q(\bdelta,\bU,V)$ can be computed independently, and as a result, determining their attendant decision rules can be carried out independently.  For this purpose, we define $V^*$ as the optimal decision rule under the attack-free setting, and $J^*_0$ as the corresponding estimation cost, i.e.,
\begin{align}\label{eq:V}
V^* \triangleq \arg\min_V J_0(V)\ ,\qquad\mbox{and}\qquad J^*_0 \triangleq \min_V J_0(V)\ .
\end{align}
\begin{definition}[$(q,\beta)$-security]\label{def1}
An estimation procedure specified by $(\bdelta,\bU,V^*)$ for the adversarial scenario is said to be $(q,\beta)$-secure if the decision rules $(\bdelta,\bU)$ yield the minimal EDF  among all the decision rules corresponding to which the average rate of missing the attacks does not exceed $\beta \in (0,1]$, i.e.,
\begin{align}\label{eq:q_opt}
q \triangleq \min_{\bdelta,\bU}q(\bdelta,\bU,V^*) \ , \qquad \mbox{\rm s.t.}\qquad \Pc_{\sf md}(\bdelta) \leq \beta \ . 
\end{align}
\end{definition}
The performance region, and its boundary that specifies the interplay between $q$ and $\beta$ are illustrated in Fig.~\ref{fig:PR}. Based on these definitions, we aim to characterize:

\begin{figure}[t]
\centering
\includegraphics[width=3 in]{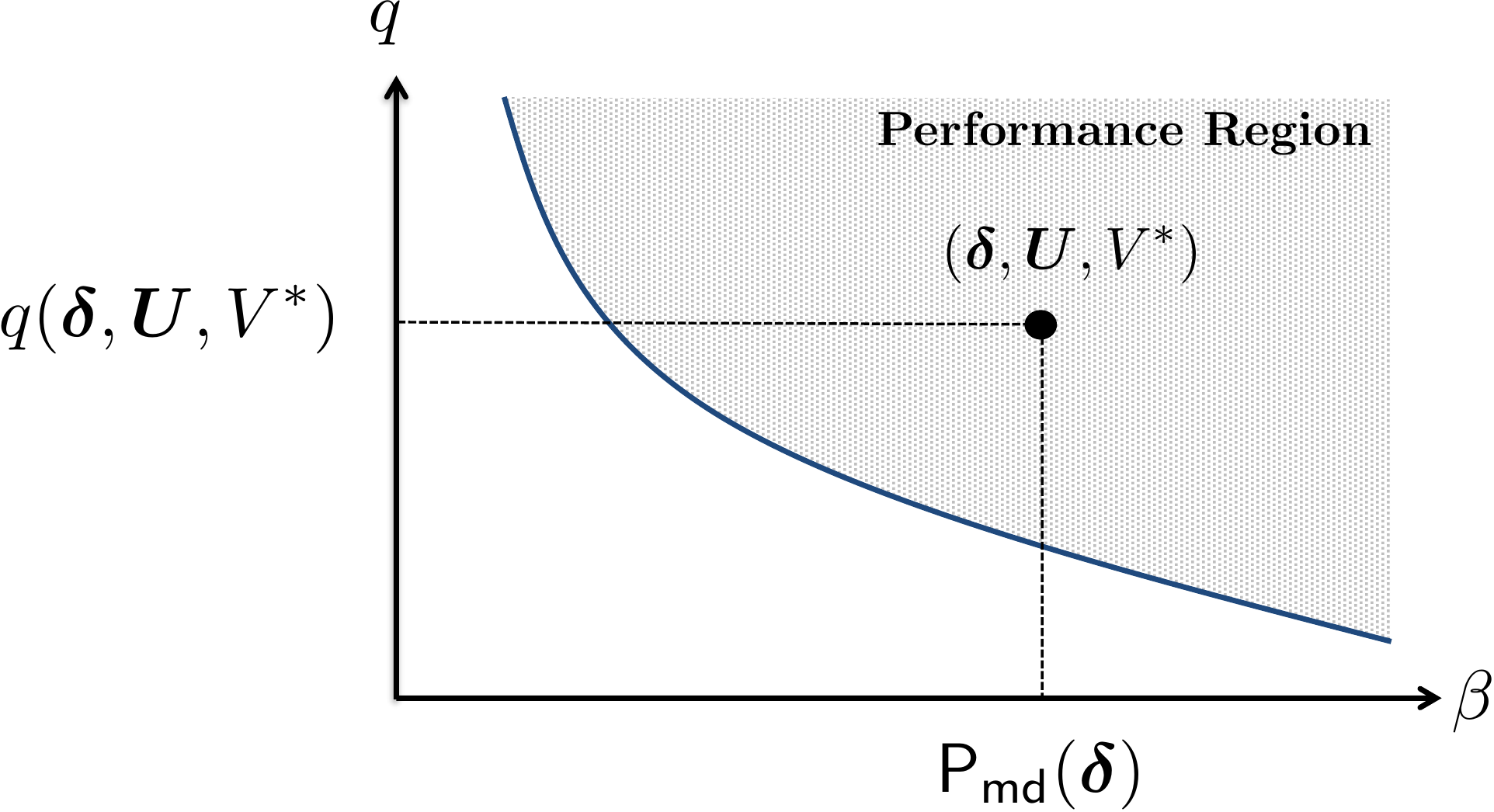}
\caption{Performance region.}
\label{fig:PR}
\end{figure}

\begin{enumerate}
\item  The region of all simultaneously achievable values of $q(\bdelta,\bU,V^*)$ and $\Pc_{\sf md}(\bdelta)$, which is illustrated by the dashed region in Fig.~\ref{fig:PR}.
\item The $(q,\beta)$-secure decision rules $(\bdelta,\bU,V^*)$ that solve~\eqref{eq:q_opt}, and specify the boundary of the performance region, which is illustrated by a solid line as the boundary of the performance region in Fig.~\ref{fig:PR}.
\end{enumerate}
By noting that $q(\bdelta,\bU,V^*)=\frac{J(\bdelta,\bU)}{J^*_0}$, where $J^*_0$ is a constant, the performance region and the $(q,\beta)$-secure decision rules are found as the solutions to
\begin{align}\label{eq:pr1}
\Q(\beta) &\triangleq  \begin{cases}
 \min_{\bdelta,\bU} & J(\bdelta,\bU)\\
\text{s.t.} & \Pc_{\sf md}(\bdelta) \leq \beta
\end{cases}\ .
\end{align}
Solving $\Q(\beta)$ ensures that the likelihood of missing an attack is confined below $\beta$. However, it is insensitive to the rate of the false alarms that the decision rules can generate. In case the statistician wishes to also control the rate of false alarms, that is the rate of erroneously declaring an attack while there is no attack, we can further extend the notion of $(q,\beta)$-security as follows. 
\begin{definition}
An estimation procedure is $(q,\alpha,\beta)$-secure if it is $(q,\beta)$-secure and the likelihood of false alarms does not exceed $\alpha\in(0,1]$. 
\end{definition}
The optimal decision rules that yield $(q,\alpha,\beta)$-secure decisions can be found as the solution to
\begin{align}\label{eq:pr2}
\Pt(\alpha,\beta) &= \begin{cases}
 \min_{\bdelta,\bU} & J(\bdelta,\bU)\\
\text{s.t.} & \Pc_{\sf md}(\bdelta) \leq \beta\\
& \Pc_{\sf fa} (\bdelta) \leq \alpha
\end{cases}\ .
\end{align}
\begin{remark}
It can be easily verified that $\Q(\beta)  = \Pt(1,\beta)$.
\end{remark}
\begin{remark}[Feasibility] \label{feasibility}
The probabilities $\Pc_{\sf md}(\bdelta)$ and $\Pc_{\sf fa}(\bdelta)$ cannot be made arbitrarily small simultaneously. Specifically, from the Neyman-Pearson theory \cite{poor2013}, it can be readily verified that for any given $\alpha$, there exists a value $\beta^{\ast}(\alpha)$, which specifies the smallest feasible value for $\beta$.
\end{remark}
We characterize the optimal solution to problems $\Pt(\alpha,\beta)$ and $\Q(\beta)$ in closed-forms in Section~\ref{sec:rules}. Close scrutiny of the optimal decision rules indicates that the complexity of the rules grows exponentially with the dimension of $X$, i.e., $p$, and the number of coordinates that an adversary can compromise, i.e. $K$. To address the computational complexity, we address the scalability issue in Section~\ref{sec:complexity}. Specifically,  we provide alternative low-complexity decision rules and show that despite their simple structure, they satisfy asymptotic optimality guarantees.

\section{Secure Parameter Estimation: Optimal Decision Rules}
\label{sec:rules}

In this section, we characterize an optimal solution to the more general problem $\Pt(\alpha,\beta)$, i.e., the estimators $\lbrace \hat X_i(\Y): i \in \left\lbrace 0,\dots,T\right\rbrace\rbrace$ and the detectors $\left\lbrace \delta_i(\Y) : i \in \left\lbrace 0,\dots,T\right\rbrace\right\rbrace$. We will also specify how the decision rules simplify for characterizing the solution to the problem $\Q(\beta)  = \Pt(1,\beta)$. In order to proceed, we start by providing the expansion of the decision error probability terms $\Pc_{\sf md}(\bdelta)$ and $\Pc_{\sf fa}(\bdelta)$ in terms of the data models and decision rules. By noting  (\ref{e_1}) and leveraging (\ref{eq:md1}), we have 
\begin{align}\label{eq:cns2}
\Pc_{\sf md} (\bdelta)&= \sum\limits_{i=1}^{T}\frac{\epsilon_i}{1 - \epsilon_0}\sum\limits_{\underset{j\neq i}{j=0}}^{T} \displaystyle \int_{\Y}\delta_j(\Y)f_i(\Y)\;\d\Y\; .
\end{align}
Similarly, by noting (\ref{e_1})
and based on (\ref{eq:fa1}), we have 
\begin{align}\label{eq:cns3}
\Pc_{\sf fa}(\bdelta) &=\sum\limits_{i=1}^{T}\displaystyle \int_{\Y}\delta_i(\Y)f_0(\Y)\; \d\Y \;.
\end{align}
By using the expansions in (\ref{eq:cns2}) and (\ref{eq:cns3}), the problem of interest in (\ref{eq:pr2}) can be equivalently cast as
\begin{align}\label{eq:prob}
\Pt(\alpha,\beta) &= \begin{cases}
\min_{(\boldsymbol{ \delta, U})} & \J(\bdelta, \bU)\\
\text{s.t.} & \sum\limits_{i=1}^{T}\frac{\epsilon_i}{1 - \epsilon_0}\sum\limits_{\underset{j\neq i}{j=0}}^{T} \displaystyle \int_{\Y}\delta_j(\Y)f_i(\Y)\;\d\Y \leq \beta\\
& \sum\limits_{i=1}^{T}\displaystyle \int_{\Y}\delta_i(\Y)f_0(\Y)\; \d\Y \leq \alpha
\end{cases}\ .
\end{align}
The roles of the estimators $\left\lbrace U_i(\Y): i \in \left\lbrace 0,\dots,T\right\rbrace\right\rbrace$ appear only in the utility function $\J(\bdelta,\bU)$. This allows for decoupling the optimization problem $\Pt(\alpha,\beta)$ into two sub-problems, as formalized in Theorem~\ref{theorem:decouple}.
\begin{theorem}\label{theorem:decouple}
The optimal secure estimators of $X$ under different models, i.e., $\hat\X=[\hat X_0,\dots, \hat X_T]$ is the solution to
\begin{align}\label{eq:prob4}
\hat \X =  \arg \min_{\bU}\J(\bdelta,\bU)\; .
\end{align}
Furthermore, the solution of $\Pt(\alpha,\beta)$, and subsequently the design of the attack detectors, can be found  by equivalently solving 
\begin{align}\label{eq:prob2}
\Pt(\alpha,\beta) &= \begin{cases}
\min_{\boldsymbol{ \delta}} & \J(\boldsymbol{\delta,\hat \X})\\
\text{\rm s.t.} & \sum\limits_{i=1}^{T}\frac{\epsilon_i}{1 - \epsilon_0}\sum\limits_{\underset{j\neq i}{j=0}}^{T} \displaystyle \int_{\Y}\delta_j(\Y)f_i(\Y)\;\d\Y \leq \beta\\
& \sum\limits_{i=1}^{T}\displaystyle \int_{\Y}\delta_i(\Y)f_0(\Y)\; \d\Y \leq \alpha
\end{cases}\ .
\end{align}
\end{theorem}
\noindent
This theorem establishes  a property for the optimal estimator  in~\eqref{eq:prob4}. By leveraging this property, and also taking into account the decoupled structure of the problem $\Pt(\alpha,\beta)$ in~\eqref{eq:prob2}, in the following theorem we provide optimal designs for the secure estimators. Interestingly, it is shown that the optimal estimator under each model can be specified by optimizing a relevant cost function defined exclusively for that model.
\newpage
\begin{theorem}[$(q,\alpha,\beta)$-secure Estimators] \label{theorem:estimator}
For the optimal secure estimators $\hat\X$ we have
\begin{enumerate}
\item The minimizer of the estimation cost $J_i(\delta_i, U_i(\Y))$, i.e., the estimation cost function under model $\H_i$, is given by
\begin{align}\label{eq:p1}
U^*_i(\Y) \triangleq \arg \inf_{U_i(\Y)} \C_{{\rm p},i}(U_i(\Y) \mid\Y)\;.
\end{align}
\item The optimal estimator $\hat\X=[\hat X_0,\dots, \hat X_T]$, specified in~\eqref{eq:prob4}, is given by
\begin{align}\label{eq:p11}
\hat X_i(\Y) = U^*_i(\Y) \ .
\end{align}
\item The cost function $\J(\bdelta, \hat\X)$ is given by 
\begin{align}\label{eq:p2}
\J(\bdelta, \hat \X) = \max_i \left\lbrace\frac{\displaystyle \int_{\Y} \delta_i(\Y){\C}_{{\rm p},i}^{\ast}(\Y)f_i(\Y) d\Y }{\displaystyle \int_{\Y}\delta_i(\Y)f_i(\Y)d\Y}\right\rbrace\; ,
\end{align}
where we have defined
\begin{align}\label{costs}
\C^*_{{\rm p},i}(\Y) \triangleq \inf_{U_i(\Y)} \C_{{\rm p},i}(U_i(\Y) \mid\Y)\; .
\end{align}
\end{enumerate}
\end{theorem}
\begin{proof}
See Appendix \ref{proof:theorem:estimator}.
\end{proof}
For illustration purposes, in the next corollary, we provide the closed-forms of these decision rules when the distributions $\{f_i(\cdot\med X): i\in\{0,\dots, T\}\}$  are all Gaussian. 
\begin{corollary}[$(q,\alpha,\beta)$-secure Estimators in Gaussian Models]
When the data models are Gaussian such that
\begin{align}\label{eq:f_gaussian}
f_i(\cdot \med X) \sim {\cal N}(\theta_i, X)\; ,\qquad \mbox{for} \quad \theta_i\in\mathbb{R}\ ,
\end{align}
where the mean values are distinct, and 
\begin{align}
X  \sim {\cal X}^{-1}(\zeta,\phi)\;,
\end{align}
where $ {\cal X}^{-1}(\zeta,\phi)$ denotes the inverse chi-squared distribution with parameters $\zeta$ and $\phi$, s.t., $\zeta+n >4$,
for the optimal secure estimators $\hat\X$, we have:
\begin{enumerate}
\item The minimizer of the estimation cost $J(\delta_i, U_i(\Y))$, i.e., the estimation cost function under model $\H_i$, is given by
\begin{align}\label{eq:p1_gaussian}
U_i^{\ast}(\Y)  = \frac{\zeta\phi + \sum\limits_{r=1}^n\| Y_r- \theta_i\|^2_2 }{\zeta + n -2}\;.
\end{align}
\item The optimal estimator $\hat\X=[\hat X_0,\dots, \hat X_T]$, specified in~\eqref{eq:prob4}, is given by
\begin{align}\label{eq:p1_gaussian}
\hat X_i(\Y) = U^*_i(\Y) \ .
\end{align}
\item The cost function $\J(\bdelta, \hat\X)$ is given by 
\begin{align}\label{eq:p2_gaussian}
\J(\bdelta, \hat \X) = \max_i \left\lbrace\frac{\displaystyle \int_{\Y} \delta_i(\Y){\C}_{{\rm p},i}^{\ast}(\Y)f_i(\Y) d\Y }{\displaystyle \int_{\Y}\delta_i(\Y)f_i(\Y)d\Y}\right\rbrace\; ,
\end{align}
where we have
\begin{align}\label{eq:p3_gaussian}
\C^*_{{\rm p},i}(\Y)  = \frac{2(\zeta\phi + \sum\limits_{r=1}^n\|Y_r-\theta_1\|^2)^2}{(\zeta_i + n - 2)^2(\zeta+n-4)}\ .
\end{align}

\end{enumerate}
\end{corollary}
Next, given the optimal estimators $\hat{\X}$, we characterize the optimal detection rules in the next theorem. The main observation is that even though we started by considering general randomized decision rules, these rules in their optimal forms reduce to deterministic ones. Furthermore, the decisions rules depend on the estimation costs that are computed based on the optimal estimation costs. These estimation costs make the decisions coupled. In order to proceed, we first show that problem $\Pt(\alpha,\beta)$ in~\eqref{eq:prob2} can be solved by leveraging the result of the following theorem, which specifies an auxiliary convex problem in a variational form.
\begin{theorem} \label{theorem:convex}
For any arbitrary $u\in\mathbb{R}_+$, we have $\Pt(\alpha,\beta)\leq u$ if and only if $\R(\alpha,\beta,u)\leq 0$, where we have defined
\begin{align}\label{pr2}
\R(\alpha,\beta,u) &= \begin{cases}
\min_{\boldsymbol{\delta}} &\gamma \\ \vspace{2pt}
\text{\rm s.t.} & \displaystyle \int_{\Y} \delta_i(\Y) f_i(\Y) [{\C}^*_{{\rm p},i}(\Y) - u]\; \d\Y \leq \gamma, \;\;  \forall \text{ }i\in \left\lbrace 0,\dots,T\right\rbrace\\ 
& \displaystyle \sum\limits_{i=1}^T \frac{\epsilon_i}{1-\epsilon_0} \sum\limits_{\underset{j\neq i}{j=0}}^T  \int_{\Y}\delta_j(\Y)f_i(\Y)\; \d\Y\  \leq \beta+\gamma\\
& \displaystyle \sum\limits_{i=1}^{T} \int_{\Y}\delta_i(\Y)f_0(\Y)\; \d\Y \leq \alpha+\gamma
\end{cases}\ .
\end{align}
Furthermore, $\R(\alpha,\beta,u)$ is convex, and $\R(\alpha,\beta,u)=0$ has a unique solution in $u$, which we denote by $u^*$. 
\end{theorem}
\begin{proof}
See Appendix~\ref{proof:theorem:convex}.
\end{proof}
The point  $u^*$ has a pivotal role in specifying the optimal detection decision rules. We define the constants $\{\ell_i: i\in \left\lbrace 0,\dots,T+2\right\rbrace\}$ as the dual variables in the Lagrange function associated with the convex problem $\R(\alpha,\beta,u^*)$. Given $u^*$ and $\{\ell_i: i\in \left\lbrace 0,\dots,T+2\right\rbrace\}$, the optimal detection rules can be characterized in closed-forms, as specified in the following theorem.
\begin{theorem}[$(q,\alpha,\beta)$-secure Detection Rules]\label{theorem:detector}
The optimal decision rules for isolating the compromised coordinates are given by
\begin{align}
\delta_i(\Y) = \begin{cases}
& 1,  \quad \text{\rm if }\;\; i = i^{\ast}\\
& 0, \quad \text{\rm if }\;\; i \neq i^{\ast}
\end{cases}\;,
\end{align}
where we have defined
\begin{align}
i^{\ast} \triangleq \argmin_{i \in \left\lbrace 0,\dots,T \right\rbrace}  A_i \; .
\end{align}
Constants $\{A_0,\dots ,A_T\}$ are specified by the data models, $u^*$, and its associated Langrangian multipliers $\{\ell_i: i\in \left\lbrace 0,\dots,T+2\right\rbrace\}$. Specifically, we have
\begin{align}
A_0 &\triangleq  \ell_0 f_0(\Y) [{\C}^*_{{\rm p},0}(\Y) - u^*] + {\ell_{T+1}} \sum\limits_{i=1}^{T} \frac{\epsilon_i}{1-\epsilon_0} f_i(\Y)\;,
\end{align}
and for $i\in \left\lbrace 1,\dots,T \right\rbrace$ we have
\begin{align}
A_i &\triangleq \ell_i f_i(\Y) [{\C}^*_{{\rm p},i}(\Y) - u^*] + \ell_{T+1}\sum_{j=1, j\neq i}^{T}\frac{\epsilon_j}{1-\epsilon_0} f_j(\Y) +  \ell_{T+2}f_0(\Y)\;.
\end{align}
\end{theorem}
\begin{proof}
See Appendix \ref{proof:theorem:detector}.
\end{proof}
In the next corollary, we provide the closed-forms of these decision rules when the distributions $\{f_i(\cdot\med X): i\in\{0,\dots, T\}\}$  are all Gaussian. 
\begin{corollary}[$(q,\alpha,\beta)$-secure Detection Rules in Gaussian Models]
When the data models $\{f_i(\cdot\med X): i\in\{0,\dots, T\}\}$ have the following Gaussian distributions
\begin{align}\label{eq:f_gaussian}
f_i(\cdot \med X) \sim {\cal N}(\theta_i, X)\; ,\qquad \mbox{for} \quad \theta_i\in\mathbb{R}\ ,
\end{align}
where the mean values are distinct,  and
\begin{align}
X  \sim {\cal X}^{-1}(\zeta,\phi)\;,
\end{align}
the optimal decision rules for isolating the compromised coordinates are given by
\begin{align}
\delta_i(\Y) = \begin{cases}
& 1,  \quad \text{\rm if }\;\; i = i^{\ast}\\
& 0, \quad \text{\rm if }\;\; i \neq i^{\ast}
\end{cases}\;,
\end{align}
where we have defined
\begin{align}
i^{\ast} \triangleq \argmin_{i \in \left\lbrace 0,\dots,T \right\rbrace}  A_i \; .
\end{align}
Constants $\{A_0,\dots ,A_T\}$ are specified by the data models, $u^*$, and its associated Langrangian multipliers $\{\ell_i: i\in \left\lbrace 0,\dots,T+2\right\rbrace\}$. Specifically, we have
\begin{align}
A_0 &\triangleq  \ell_0 f_0(\Y) ({\C}^*_{{\rm p},0}(\Y) - u^*) + {\ell_{T+1}} \sum\limits_{i=1}^{T} \frac{\epsilon_i}{1-\epsilon_0} f_i(\Y)\;,
\end{align}
and for $i\in \left\lbrace 1,\dots,T \right\rbrace$ we have
\begin{align}
A_i &\triangleq \ell_i f_i(\Y) ({\C}^*_{{\rm p},i}(\Y) - u^*) + \ell_{T+1}\sum_{\underset{j\neq i}{j=1}}^{T}\frac{\epsilon_j}{1-\epsilon_0} f_j(\Y) +  \ell_{T+2}f_0(\Y)\;.
\end{align}
where ${\C}^*_{{\rm p},i}(\Y)$ is evaluated using \eqref{eq:p3_gaussian} and we have
\begin{align}
f_i(\Y) = \frac{(\zeta \phi )^{\frac{\zeta}{2}}}{\pi^{\frac{n}{2}} \Gamma(\zeta/2)} \cdot \frac{\Gamma(\zeta + n)/2}{( \zeta\phi + \sum\limits_{r=1}^n \|Y_r - \theta_i\|^2 )^{\frac{\zeta + n}{2}}}\;.
\end{align}
\end{corollary}
By setting $T= 1$, $n = 1$,  $\theta_0=0$, $\theta_1 = 2$, and selecting $\zeta = 4, \phi = 1$, Fig.~\ref{fig:PRcurve} depicts the performance region and the associated $(q,\beta)$-security curve, which shows the tradeoff between the quality of the detection and the degradation in the estimation quality. It is noteworthy that this tradeoff is inherently due to the formulation of the secure estimation problem. Essentially, problem $\Pt(\alpha,\beta)$ as specified in~\eqref{eq:pr2}, is designed to trade the quality of detection in favor of improving the estimation cost.

\begin{figure}[t]
\centering
\includegraphics[width=3 in]{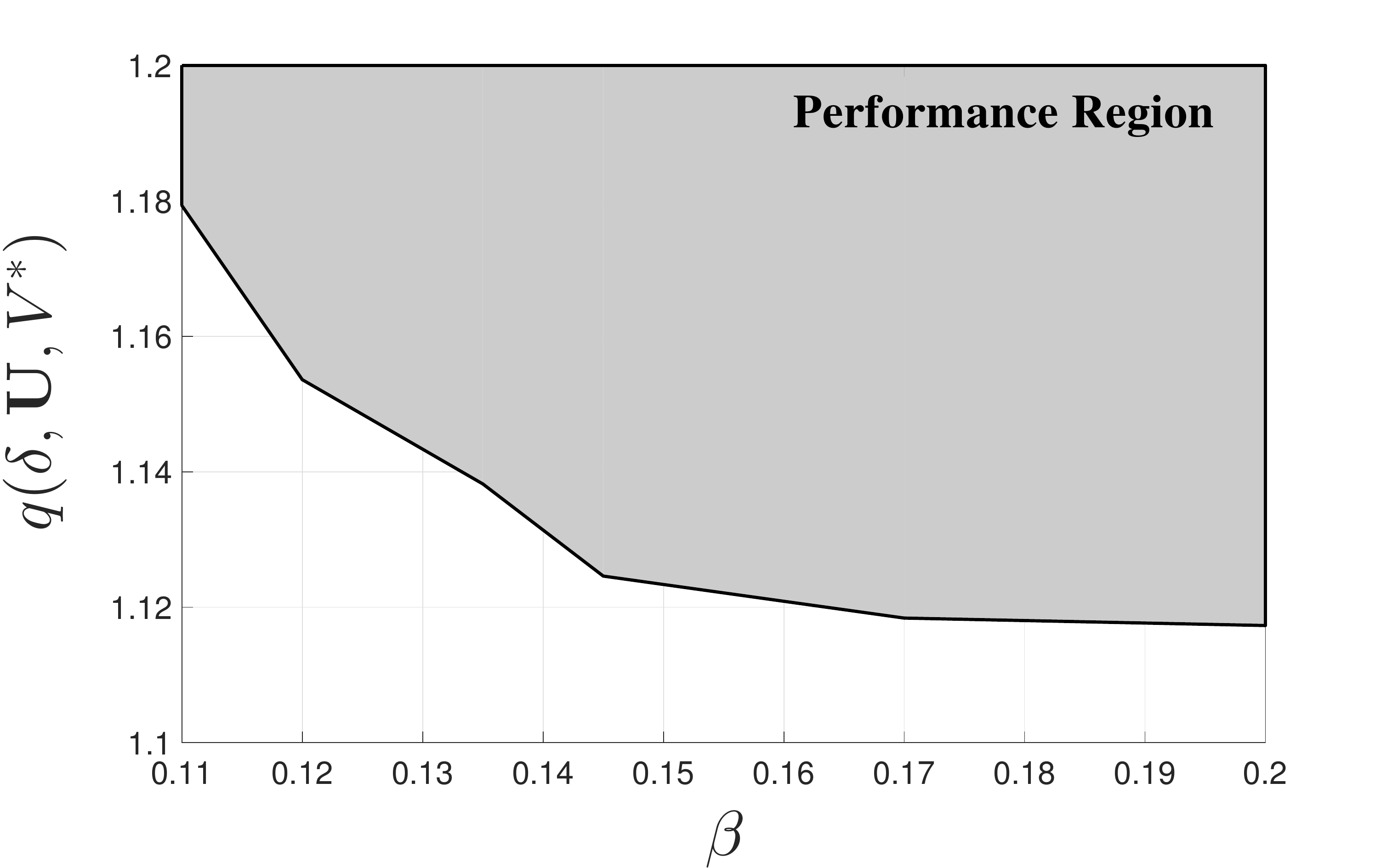}
\caption{Performance region for the Gaussian data model.}
\label{fig:PRcurve}
\end{figure}

Based on all the decision rules specified in the section, and the detailed steps of specifying the parameters involved in characterizing the decision rules, we provide Algorithm~1, which summarizes all the steps for solving $\Pt(\alpha,\beta)$ for any feasible pair of $\alpha$ and $\beta$.


\begin{algorithm}
\caption{-- Solving $\Pt(\alpha,\beta)$}
\label{array-sum}
\begin{algorithmic}[1]
\State \text{input } $\alpha$ and $\beta$ and evaluate $\beta^{\ast}(\alpha)$
\If {$\beta < \beta^{\ast}(\alpha)$}
\State \text{$\Pt(\alpha,\beta)$ not feasible for given choice of $\alpha$ and $\beta$}
\State \text{break}
\Else 
\State \text{Initialize }$u_0 = 0, u_1$
\State \text{Evaluate optimal posterior estimation costs in \eqref{costs}}
\Repeat
    \State $\hat{u} \gets (u_0 + u_1)/2$
    \For {every $\boldsymbol{\hat{\ell}}\succcurlyeq 0$ in the discretized space $\|\boldsymbol{\hat{\ell}}\|_1 = 1$}
    \State \text{Compute $\boldsymbol{\delta}$ from Theorem \ref{theorem:detector}}
    \State \text{Compute $M(\boldsymbol{\hat{\ell}}) \triangleq \R(\alpha,\beta,\hat{u}) $}
    \EndFor
    \If {$\min_{\boldsymbol{\hat{\ell}}}M(\boldsymbol{\hat{\ell}}) \leq 0$}
    \State $u_1 \gets \hat{u}$
    \State $\boldsymbol{\ell} \gets \boldsymbol{\hat{\ell}}$
    \Else
    \State $u_0 \gets \hat{u}$
    \EndIf
	
\Until $u_1 - u_0 \leq \epsilon$\text{, for }$\epsilon$\text{ sufficiently small}
\State $\Pt(\alpha,\beta)\gets u^{\ast} = u_1$\\
\Return {\text{Decision rules $\boldsymbol{\delta}$ }}
\EndIf
\end{algorithmic}
\end{algorithm}

\newpage
\section{Scalable Secure Parameter Estimation }\label{sec:complexity}
As the data dimension $m$ grows, the number of possible data models $T$ grows exponentially. This, subsequently, leads to an exponential growth in the complexity of forming the decision rules, e.g., the number of Lagrangian multipliers needed for characterizing the detection rules scales linearly in $T$. This can render Algorithm~1 computationally prohibitive. In order to circumvent the computational complexity, in this section, we develop a scalable approach to secure estimation, which exhibits optimality properties too. The core idea is to break down the problem into $m$ coordinate-level problems, treat them individually, and then aggregate the individual decisions. Specifically, the decisions involve a binary decision on whether the data is compromised. If the data is deemed to be compromised, then each coordinate is tested individually, and an estimate of $X$ is formed based on the data of that coordinate. Individual coordinate-level estimates are then tested for reliability, and combined to form an aggregate estimate for $X$. Since we need to perform only $m$ single-coordinate binary detection decisions followed by forming a coordinate-level estimation routine, the computational complexity scales only linearly in the number of coordinates $m$, as opposed to exponentially for forming the optimal decision rules.

\subsection{Binary Attack Detection}
\label{sec:binary}

In the first stage, we perform a binary test to detect whether data is compromised at all. This is carried out by solving a  binary composite hypothesis testing problem given by
\begin{equation}\label{eq:hyp2}
\begin{array}{ll}
\vspace{0.1in}
\hat{\H}_0\ : & \Y \sim f (\Y \mid X),\ \text{ with }X \sim \cpi(X)\\
\hat{\H}_1\ : & \Y \sim \hat f(\Y \mid X),\ \text{ with }X \sim \cpi(X)\\
\end{array} ,
\end{equation} 
where $\hat{\H}_0$ is the hypothesis corresponding to the attack-free setting, and $\hat{\H}_1$ signifies to the presence of an attack. Probability distribution $\hat f$ is a mixed distribution given by
\begin{equation}\label{eq:mixed}
\hat f(\Y) \triangleq \frac{1}{1-\epsilon_0}\sum\limits_{i=1}^{T} \epsilon_i f_i(\Y)\;.
\end{equation}
Similarly to the optimal approach of Section~\ref{sec:rules}, to design the decision rules we define the randomized test $\boldsymbol{ \hat{\delta}} \triangleq [ \hat{\delta}_0(\Y),\hat{\delta}_1(\Y)  ]$, in which $\hat{\delta}_i(\Y) \in \left[0,1\right]$ is the probability of deciding in favor of $\hat \H_i$, and we have $\delta_0(\Y) + \delta_1(\Y) = 1$.  Furthermore, define $\D_{\sf d} \in \lbrace \hat{\H}_0, \hat{\H}_1 \rbrace$ and $\T_{\sf d} \in \lbrace \hat{\H}_0, \hat{\H}_1\rbrace$ as the decision and the true hypotheses. Hence, the false alarm rate is given by 
\begin{equation}\label{eq:fa2}
\P (\D_{\sf d} = \hat{\H}_1 \mid \T_{\sf d} = \hat{\H}_0) = \int \hat{\delta}_1(\Y) f(\Y) \;\d\Y \ ,
\end{equation}
and the miss-detection rate is given by
\begin{equation}\label{eq:md2}
\P (\D_{\sf d} = \hat{\H}_0 \mid \T_{\sf d} = \hat{\H}_1) = \int \hat{\delta}_0(\Y) \hat f(\Y) \;\d\Y \;.
\end{equation}
Since our aim is to maximize the true detection rate (or minimize (\ref{eq:md2})) in the first step, we design the decision rule $\boldsymbol{\hat{\delta}}$ using the Neyman Pearson approach. Following the Neyman-Pearson approach, as noted in Remark~\ref{feasibility}, we have the following decision rules for the binary detection decision~ \cite{poor2013}.
\begin{theorem}[Attack Detection]\label{theorem:NP}
The optimal attack detection rule that minimized the rate of missing the attacks subject to a constraint on the false alarms is given by 
\begin{align}
\label{eq:rtest1} \hat{\delta}_0(\Y) & = \mathbbm{1}_{\left\lbrace \frac{\hat f(\Y)}{f(\Y)} < \gamma \right\rbrace} + (1-\varrho)\cdot \mathbbm{1}_{\left\lbrace \frac{\hat f(\Y)}{f(\Y)} = \gamma \right\rbrace}  \ , \\
\label{eq:rtest2} \hat{\delta}_1(\Y) & = \mathbbm{1}_{\left\lbrace \frac{\hat f(\Y)}{f(\Y)} > \gamma \right\rbrace} + \varrho\cdot \mathbbm{1}_{\left\lbrace \frac{\hat f(\Y)}{f(\Y)} = \gamma \right\rbrace} \; \ .
\end{align}
where the threshold $\gamma$ and the probability term $\varrho$ are chosen such that the false alarm constraint is satisfied with equality. 
\end{theorem}

\subsection{Isolating Compromised Coordinates}
\label{sec:isolation}

The outcome of the binary detection rule in Section~\ref{sec:binary} is ruling whether there exists an attack in the data. If an attack is deemed to exists, in the second step, we carry out an isolation decision, the role of which is identifying the compromised coordinates. In this subsection, we start by analyzing the optimal isolation rules when an attack is deemed to exist in Section~\ref{sec:isolationperformance}, and characterize the performance of the optimal decision rules. This analysis will serve as a baseline for comparing the performance of any alternative low-complexity isolation rule. Next, we provide an isolation rule in Section~\ref{sec:asymptoticisolation} that has low-complexity and can achieve the same performance asymptotically as that of the optimal decision rule.

\subsubsection{Optimal Isolation Rule}
\label{sec:isolationperformance}

If the attack detection rules specified by~\eqref{eq:rtest1}-\eqref{eq:rtest2} determine that the data in one or more coordinates are compromised, in the next step we aim to identify the compromised coordinates. Isolating the  set of compromised coordinates is equivalent to solving the following $T$-hypothesis testing problem.
\begin{equation}\label{eq:hyp3}
\begin{array}{ll}
\vspace{0.1in}
\H_i\ : & \Y \sim f_i(\Y \mid X),\ \text{ with }X \sim \pi(X)\ , \;\quad \mbox{for}\;\;i\in\{1,\dots, T\}\ .
\end{array} 
\end{equation} 
Define $\D_{\sf is}\in \{\H_1,\dots,\H_T\}$ and $\T_{\sf is} \in \{\H_1,\dots,\H_T\}$ as the decision formed and the true hypothesis, respectively.  We define the randomized test $\boldsymbol{ \hat{\delta}}_1 (\Y) \triangleq [ \hat{\delta}_{11}(\Y), \dots, \hat{\delta}_{1T}(\Y)  ]$, in which $\hat{\delta}_{1i}(\Y) \in \left[0,1\right]$ is the probability of deciding in favor of $\hat \H_i$ given that we have already decided that part of the data is compromised, i.e., $\hat\H_1$ is the true hypothesis in~\eqref{eq:hyp2}. Accordingly, we define $\Pc_{\sf is}(\boldsymbol{ \hat{\delta}}_1)$ as the probability of making an erroneous  decision on the problem in \eqref{eq:hyp3}, given that the decision in the detection step was $\hat{\H}_1$. Therefore, $\Pc_{\sf is}(\boldsymbol{ \hat{\delta}}_1)$ is given by 
\begin{align}\label{eq:Pis}
\Pc_{\sf is} (\boldsymbol{ \hat{\delta}}_1) &\triangleq \P(\D_{\sf is} \neq \T_{\sf is} \mid \D_{\sf d} = \hat\H_1)\\
&= \sum\limits_{i = 1}^T \sum\limits_{j = 1, j\neq i}^T \P(\D_{\sf is} = \H_j \mid \T_{\sf is} = \H_i,\; \D_{\sf d} = \hat\H_1)\cdot \P(\T_{\sf is} = \H_i\med \D_{\sf d} = \hat\H_1)\; .
\end{align}
We also denote the error exponent of $\Pc_{\sf is} (\boldsymbol{ \hat{\delta}}_1)$ as the number of the samples $n$ grows according to
\begin{align}
\psi(\boldsymbol{ \hat{\delta}}_1) \triangleq - \lim_{n\rightarrow \infty}\frac{\log \Pc_{\sf is} (\boldsymbol{ \hat{\delta}}_1)}{n}\ .
\end{align}
In the next theorem, we characterize the optimal decision rule $(\boldsymbol{ \hat{\delta}}_1)$ and the associated error exponent. For this purpose, we denote the Chernoff information between two probability measures with probability density functions $g$ and $h$ by
\begin{align}\label{eq:Chernoff}
C(g,h)\triangleq -\log \min_{\alpha \in (0,1)} \int g^\alpha (x) h^{1-\alpha}(x) \; \d x\ .
\end{align}
\begin{theorem}\label{lemma:errorprob}
The decision rule $\boldsymbol{ \hat{\delta}}_1$ that minimizes $\Pc_{\sf is} (\boldsymbol{ \hat{\delta}}_1)$ is given by
\begin{align}\label{drule1}
\hat \delta_{1i}(\Y) = \begin{cases}
& 1,  \quad \text{\rm if }\;\; i = i^{\ast}\\
& 0, \quad \text{\rm if }\;\; i \neq i^{\ast}
\end{cases}\;,
\end{align}
where
\begin{align}
i^{\ast} = \argmax_{j \in \{1,\dots,T\}} f_j(\Y)\P(\T_{\sf is}= \H_j \mid \T_{\sf d} = \hat\H_1)\;. 
\end{align} 
\end{theorem}
\begin{proof}
See  Appendix \ref{proof:lemma:errorprob}.
\end{proof}
Note that the optimal decision rules will have the same computational complexity as the optimal rules characterized in Theorem~\ref{theorem:detector}.

 \subsubsection{Asymptotically Optimal Isolation Rule}
 \label{sec:asymptoticisolation}
 
The optimal decision rule in Theorem~\ref{lemma:errorprob} has similar drawbacks in terms of computational complexity in high dimensions as the optimal decision rules discussed in Section~\ref{sec:rules}. In this subsection, we provide an alternative isolation rule for discerning the compromised coordinates at a lower computational complexity, and show that this rule is optimal in the asymptote of large number of samples, i.e.,  $n\rightarrow\infty$.  We denote the alternative low-complexity decision rule by  $\boldsymbol{\bar\delta_1}$ and start by providing lower and upper bounds on $\Pc_{\sf is}(\boldsymbol{\bar{\delta}_1})$, and show that these bounds exhibit the same error exponents. To specify a lower bound we first leverage the fact that $\Pc_{\sf is}(\boldsymbol{\hat{\delta}_1}) \leq \Pc_{\sf is} (\boldsymbol{ \bar{\delta}}_1)$, which is true due to the optimality of $\Pc_{\sf is}(\boldsymbol{\hat{\delta}_1})$, and define $\Pc_{\sf is}(X,(\boldsymbol{ \bar{\delta}}_1))$ as the value of $\Pc_{\sf is} (\boldsymbol{ \bar{\delta}}_1)$ when the unknown parameter $X$ is 
fully known (e.g., provided by a genie).  Clearly, when $X$ is fully known, the error probability decreases, as part of the model uncertainty due to the unknown parameter is removed, and the decision problem reduces to a purely detection problem. To specify the upper bound on $\Pc_{\sf is} (\boldsymbol{ \bar{\delta}}_1)$, we  define  $\Pc_{\sf is} (X_c,\boldsymbol{ \bar{\delta}}_1)$ as the value of $\Pc_{\sf is} (\boldsymbol{ \bar{\delta}}_1)$ when $X$ is
assumed to take an arbitrary value $ X_c\in\mathbb{R}^p$.  Note that replacing $X$ by arbitrarily chosen $X_c$ induces sub-optimality compared to the optimal case in which optimal estimation is performed. The following lemma specifies  the bounds on $\Pc_{\sf is} (\boldsymbol{ \bar{\delta}}_1)$.

\begin{lemma}\label{it}
There exists some $X_c  \in \mathbb{R}^{p}$, such that
\begin{align}
\Pc_{\sf is}^{l}\triangleq \Pc_{\sf is}(X,(\boldsymbol{ \hat{\delta}}_1))  \leq \Pc_{\sf is}(\boldsymbol{ \hat{\delta}}_1) \leq \Pc_{\sf is}(\boldsymbol{ \bar{\delta}}_1)  \leq  \Pc_{\sf is}^{u}\triangleq  \Pc_{\sf is}(X_c,(\boldsymbol{ \bar{\delta}}_1)) \ .
\end{align}
\end{lemma}
This lemma  is instrumental in establishing the error exponent of the alternative low-complexity decision rule that we will provide. To proceed, corresponding to each coordinate $l\in\{1,\dots,m\}$ we define the likelihood ratio term 
\begin{align}\label{eq:lr}
\LR_l (\Y)\triangleq   \prod_{r=1}^n \frac{g_l^1(Y_r(l))}{g_l^0(Y_r(l))}\; ,
\end{align}
where $g_l^0$ and $g_l^1$ denote the marginal pdfs of the data at coordinate $l\in\{1,\dots,m\}$ under the attack-free setting and when the coordinate is compromised, respectively. In the next theorem, we show that a decision rule based on calculating these likelihood ratio terms suffices to reach a decision that achieves the same error exponent as the optimal decision rule specified in~\eqref{eq:Pis}.

\begin{theorem}\label{theorem:asympIsolation}
The isolation rule
\begin{align}\label{drule1}
\bar \delta_{1i}(\Y) = \begin{cases}
& 1,  \quad \text{\rm if }\;\; i = i^{\ast}\\
& 0, \quad \text{\rm if }\;\; i \neq i^{\ast}
\end{cases}\; ,
\end{align}
in which we have defined
\begin{align}
i^{\ast} = \argmax_{i \in \{1,\dots T\}} \prod\limits_{v\in S_i} \LR_v(\Y)\;,
\end{align}
has the following error exponent
\begin{align}
\psi(\boldsymbol{ \bar{\delta}}_1) = \min_{i\neq j \in \{1,\dots,T\}} C(f_i,f_j)\; .
\end{align}
which is equal to $\psi(\boldsymbol{ \hat{\delta}}_1)$, i.e., the error exponent of the optimal rule.
\end{theorem}
\begin{proof}
See Appendix \ref{proof:theorem:isolator}.
\end{proof}
Therefore, the error probabilities corresponding to the optimal decision rule in Theorem~\ref{lemma:errorprob} and the low-complexity alternative rule Theorem~\ref{theorem:asympIsolation} decay at the same rate with the increasing number of samples $n$, rendering the decision rule in Theorem~\ref{theorem:asympIsolation} asymptotically optimal. Clearly, the decision rule based on the marginal likelihood ratios significantly reduces the computational complexity for isolating the attacked coordinates. This reduction in computational complexity is viable at the expense of accuracy of the isolation when we have a small number of samples. In the next subsection, we discuss how the attack detection and coordinate isolation decisions characterized so far are leveraged for forming an estimate of $X$.

\subsection{Coordinate-based Secure Estimation}

The decision rules in Section~\ref{sec:binary} and Section~\ref{sec:isolation} have already produced decisions about whether each individual coordinate is compromised. In the third stage of the decisions, we form one estimate of $X$ based on all the $n$ samples available at each coordinate. This leads to forming $m$ distinct estimates for $X$, one corresponding to each coordinate. Clearly, not all the estimates can be deemed reliable, especially when some of the coordinates are compromised. For this purpose, after forming the $m$ estimates we perform a {\sl reliability} test, the purpose of which is discarding the unreliable estimates and retaining and aggregating the reliable ones. We provide relevant estimation cost functions in Section~\ref{sec:costfunctions}, and characterize the reliability test in Section~\ref{sec:reliability}.

\subsubsection{Cost Functions}
\label{sec:costfunctions}

Based on the sequence of the data points at each coordinate, we form an estimate of $X$ corresponding to each coordinate, resulting in $m$ distinct estimates for $X$. Clearly, different coordinates produce estimates of $X$ with potentially different qualities. Motivated by the fact that the decisions formed in the previous steps are not perfect and a coordinate might yield an unreliable estimate for $X$, we consider performing a reliability test on the decision produced at each coordinate. For this purpose, at each coordinate $l$ and based on the data available at coordinate $l$, which we denote by $\Y_l$, we perform a binary test to decide whether the estimate based on the data from coordinate $l$ is reliable or unreliable, denoted by $\H^{\sf r}_l$ and $\H^{\sf u}_l$, respectively. 

We define $\D^{\sf r}_l \in \{\H^{\sf r}_l, \H^{\sf u}_l\}$ as the decision formed about the reliability of the estimate, and let the randomized test $\bar{\boldsymbol{\delta}}_l(\Y_l)= [\bar{\delta}_l^{\sf r}(\Y_l), \bar{\delta}^{\sf u}_l(\Y_l) ]$ be  the decision rule to decide upon the reliability of the estimate from coordinate $l$, where $\bar\delta_l^{\sf r}(\Y_l)$ is the probability of deciding in favor of $\H^{\sf r}_l$ and  $\bar\delta_l^{\sf u}(\Y_l)$ is the probability of deciding in favor of $\H^{\sf u}_l$. Furthermore, we define $\boldsymbol{\delta}_l (\Y) \triangleq  [\delta^0_l(\Y_l), \delta^1_l(\Y_l)]$, where $\delta^1_l(\Y_l)$ denotes the  probability that coordinate $l$ is being compromised based on data $\Y$ and the decision rules in Section~\ref{sec:binary} and Section~\ref{sec:isolation}, and subsequently, $\delta^0_l(\Y_l)=1-\delta^1_l(\Y_l)$. Hence, the likelihood of forming a reliable estimate at coordinate $l$ given that we have decided that coordinate $l$ is not compromised is given by
\begin{align}\label{eq:P0}
\P_0(\D^{\sf r}_l = \H^{\sf r}_l) &= \int \delta^0_l(\Y_l) \bar\delta_l^{\sf r}(\Y_l) g_l^0(\Y_l) \d\Y_l\; .
\end{align}
Similarly,  the likelihood of forming a reliable estimate at coordinate $l$ given that we have decided that coordinate $l$ is compromised is given by
\begin{align}\label{eq:P1}
\P_1(\D^{\sf r}_l = \H^{\sf r}_l) & = \int \delta^1_l(\Y_l) \bar\delta_l^{\sf r}(\Y_l) g_l^1(\Y_l) \d\Y_l\; .
\end{align}
Furthermore, consider $U^0_l$ and $U^1_l$ as the estimates of $X$ at coordinate $l$ when we have decided that the coordinate is attack-free and  compromised, respectively. Hence, the cost associated with $U_l^j$ when the decision of the reliability test is $\H^{\sf r}_l$ is defined as 
\begin{align}
\J_l^j(\delta_l^j,\bar\delta_l^{\sf r},U_l^j) &\triangleq \E_j[\C(X,U_l^j) \mid \D_{ l}^{\sf r} = \H_l^{\sf r}]\; \\
&= \frac{\displaystyle\int\int \delta^j_l(\Y_l) \bar\delta^{\sf r}_l(\Y_l)\C(U_l^j,X) g_l^j(\Y_l \mid X) \cpi(X) \d X \d\Y_l}{\displaystyle \int\int \delta^j_l(\Y_l) \bar\delta^{\sf r}_l(\Y_l) g_l^j(\Y_l \mid X) \cpi(X) \d X \d\Y_l}\; \\
&= \frac{\displaystyle\int \delta_l^j(\Y_l) \bar\delta_l^{\sf r}(\Y_l)\C^j_l(U_l^j\mid \Y_l) g_l^j(\Y_l) \d\Y_l}{\displaystyle\int \delta^l_j(\Y_l) \bar\delta^{\sf r}_l(\Y_l) g_l^j(\Y_l ) \d\Y_l}\; ,
\end{align}
where $\C^j_l(U_l^j\mid \Y_l) $ is the posterior estimation cost at coordinate $l$.  We define the optimal average estimation cost as
\begin{align}
\hat\C^j_l(\Y_l) \triangleq \min_{U_l^j} \C_l^j(U_l^j\mid \Y_l)\;,
\end{align}
and the optimal coordinate-level estimators as
\begin{align}\label{eq:cest}
\hat{X}_l^j(\Y_l) \triangleq \arg\min_{U_l^j} \C_l^j(U_l^j\mid \Y_l)\;.
\end{align}

\subsubsection{Reliability Decision Rules}
\label{sec:reliability}

For the probabilities of forming a reliable estimate at coordinate $l$ defined in~\eqref{eq:P0}~and~\eqref{eq:P1} we have
\begin{align}\label{eq:con1}
\P_j(\D_{l}^{\sf r}   = \H_l^{\sf r}) &= \int \delta_l^j(\Y_l) \bar\delta_l^{\sf r}(\Y_l) g_l^j(\Y_l) \d\Y_l\;,\\
& \leq \int \delta_l^j(\Y_l)  g_l^j(\Y_l) \d\Y_l\\
& =  \rho^j_l\; ,
\end{align}
where $(1-\rho_l^0)$ and $(1-\rho_l^0)$ are the Type-I and Type-II probabilities of mis-classifying the status of coordinate~$l$. Therefore, the probability of forming a reliable estimate, when coordinate $l$ is deemed to be compromised or attack-free is upper bounded by $\rho^1_l$ and $\rho_l^0$, respectively. This implies that only a fraction of the decisions on the true model of data at coordinate $l$ will provide reliable estimates.  We wish to have decision rules that control the fraction of the reliable estimates to be beyond a pre-specified  level. Obviously, these desired levels should also be within the feasible range. For this purpose, we select $\nu_l^j\in[0,\rho_l^j]$,  and impose a lower bound on the fraction of the reliable estimates that the reliability test retains according to 
\begin{align}\label{eq:con2}
\P_j(\D_{l}^{\sf r}= \H_l^{\sf r})\geq \nu_l^j \ .
\end{align}

\noindent Based on the cost functions and decision constraints, the decision rules $\boldsymbol{\bar{\delta}}_l$ and the estimators $U_l^j$ are determined by solving $m$ problems in parallel, where the problem to be solved corresponding to coordinate $l$ is given by
\begin{align}\label{opt}
{\cal S}(\nu^j_l) \triangleq \begin{cases}\min_{\boldsymbol{\bar{\delta}}_l, U_l^j}& \J_l^j(\delta_l^j,\bar\delta_l^{\sf r},U_l^j)\\
 \text{s.t. } & \P_j(\D_{l}^{\sf r} = \H_l^{\sf r}) \geq \nu_l^j
 \end{cases} \; .
\end{align}
Similarly to the discussion in Section~\ref{sec:rules}, since the effect of the estimators $U_l^j$ appears only in the cost function $\J_l^j(\delta_l^j,\bar\delta_l^{\sf r},U_l^j)$, the optimization problem in (\ref{opt}) can be decoupled into two subproblems as formalized in the following theorem.
\begin{theorem}\label{rel}
For given $\nu_l^0$ and $\nu_l^1$, the optimal decision rules for the problems ${\cal S}(\nu_l^0)$ and ${\cal S}(\nu_l^1)$ are given by
\begin{align}
\hat{\C}^0_{l}(\Y_l) &\mathrel{\mathop{\gtrless}^{\H_l^{\sf u}}_{\mathrm{\H_l^{\sf r}}}} \gamma_l^0\quad \mbox{if coordinate $l$ is deemed attack-free}\ ,\\
\hat{\C}^1_{l}(\Y_l) &\mathrel{\mathop{\gtrless}^{\H_l^{\sf u}}_{\mathrm{\H_l^{\sf r}}}} \gamma_l^1\quad \mbox{if coordinate $l$ is deemed compromised}\;,
\end{align}
where $\gamma_l^1$ and $\gamma_l^0$ are selected such that the constraints in (\ref{opt}) are satisfied with equality. The optimal estimate that solve ${\cal S}(\nu_l^0)$ and ${\cal S}(\nu_l^1)$ are the estimators $\hat{X}_l^0(\Y_l)$ and $\hat{X}_l^1(\Y_l)$, respectively, defined in~\eqref{eq:cest}.
\end{theorem}

\begin{proof}
See Appendix \ref{Thm7}.
\end{proof}

For given data $\Y$, performing the reliability test in Theorem~\ref{rel} retains the estimates deemed reliable. These estimates, subsequently, can be aggregated to form a single estimate for $X$. The estimation approach, which consists of global binary attack detection, followed by coordinate-level isolation of compromised coordinates and local estimates, and completed by aggregating the reliable coordinate-level estimates, renders a scalable approach to the secure estimation of interest. This approach is significantly less complex as compared to the optimal estimation rules characterized in Theorem~\ref{theorem:estimator}. Finally, we remark that aggregating the local estimates is a well-investigated problem, to which there exists a broad range of solutions with varying degrees of optimality under various settings. Representative approaches to  such aggregation include the studies in~\cite{distest3,lineardec,fusion,xiao2005}. A summary of the steps involved in the scalable approach is provided in Algorithm~2.

\begin{algorithm}
\caption{-- Scalable and Asymptotically Optimal Solution to $\Pt(\alpha,\beta)$}
\label{array-sum}
\begin{algorithmic}[1]
\State \text{input feasible} $\nu_l^0$\text{ and }$\nu_l^1$\text{ for all coordinates $l$} 
\State \text{Binary attack detection}
\If {\text{Attack exists}}
\State \text{Isolate the true model using decision rule in Theorem \ref{theorem:asympIsolation}}
\EndIf
\State \text{Form estimates based on data at each coordinate}
\State\text{Test for estimate reliability at each coordinate using decision rules in Theorem \ref{rel} }
\State\text{Aggregate reliable coordinate level estimates to form the final estimate}
\end{algorithmic}
\end{algorithm}

\newpage
\section{Case Studies: Secure Estimation in Sensor Networks}
We use the example of a sensor network consisting of two sensors and a fusion center (FC) to evaluate the estimation frameworks presented in this paper. Each sensor is collecting a stream of data. Sensor $i\in\{1,2\}$ collects $n$ measurements, denoted by $\Y_i=[Y^i_1,\dots,Y^i_n]$, where each sample $Y^i_j\in\mathbb{R}$ in an attack-free environment is related to $X$ according to 
\begin{align}\label{eq:attack_free}
Y_j^i = h^iX + N_j^i\ ,
\end{align}
where $h^i$ models the channel connecting sensor $i$ to the FC and  $N_j^i$ accounts for the additive channel noise. Different noise terms are assumed to be independent and identically distributed (i.i.d.) generated according to a known distribution. We will consider two adversarial scenarios that impact the data model in~\eqref{eq:attack_free} and evaluate the optimal performance, as well as the application of the asymptotically optimal scalable algorithm when the number of the samples $n$ tends to infinity.

\subsection{Case 1: One Sensor Vulnerable to Causative Attacks}
We start by considering an adversarial setting in which the data model of the measurements from only one sensor (sensor 1) are vulnerable to a causative attack, while the other sensor (sensor 2) remains attack-free. Under such a setting, we have only one attack scenario, i.e., $T=1$ and $S_1=\{1\}$. Accordingly, we have $\epsilon_0+\epsilon_1=1$. Under the attack-free scenario, we assume that the noise terms $N_j^i$ are i.i.d. and distributed according to $\N(0,\sigma^2_n)$, i.e.,
\begin{align}
Y_j^i \mid X \sim \N(h^i X, \sigma_n^2)\ .
\end{align}
When data from sensor $1$ is compromised, the actual conditional distribution of  $Y_j^1|X$ is distinct from the above distribution assumed by the statistician. The inference objective under such a setting, in principle, becomes similar to the adversarial setting of~\cite{wilson2016}, which focuses on data injection. Hence, in order to be able to compare the performance of the optimal framework with that of~\cite{wilson2016}, we assume that the conditional distribution of  $Y_j^1|X$ when sensor 1 is under a causative attack is $\N(h^i X,\sigma^2_n) \ast {\sf Unif}[a,b]$, where $a,b\in\mathbb{R}$ are fixed constants and $\ast$ denotes convolution. Therefore, the composite hypothesis test for estimating $X$ and discerning the model in~\eqref{eq:hyp1} simplifies to the following binary test with the prior probabilities $\epsilon_0$ and $\epsilon_1$, in which we have defined $\Y=[\Y_1,\Y_2]$. 
\begin{equation}\label{eq:hyp:case}
\begin{array}{ll}
\H_0\ : & \Y \sim f_0(\Y\mid X),\ \text{ with }X \sim \N(0,\sigma^2)\\
\H_1\ : & \Y \sim f_1(\Y\mid X),\ \text{ with }X \sim \N(0,\sigma^2)\\
\end{array} ,
\end{equation}
Figure \ref{fig:1} depicts the variations of the estimation quality, captured by $q$, versus the tolerable miss-detection rate $\beta$, where it is observed that the estimation quality improves monotonically as $\beta$ increases, and it reaches its maximum quality when $\beta=1$. 
This observation is in line with what is expected analytically from the formulation of the secure parameter estimation problems in \eqref{eq:pr1} and \eqref{eq:pr2}.

A similar setting is studied in~\cite{wilson2016}, where the attack is induced additively into the data of sensor 1 and can be any real number. This setting can be studied in the context of causative attacks where the attacker's mode of compromising the data is adding a disturbance which has uniform distribution. Figure~\ref{fig:1} also compares the estimation quality of the methodology developed in this paper, with that obtained by applying the methodology of~\cite{wilson2016}, which characterizes a single point in the $(q,\beta)$ plane. Specifically, in~\cite{wilson2016}, an estimator is designed to obtain the most robust estimate by exploring the dependence of the estimation quality on the false alarm probability, using which an optimal false alarm probability $\alpha^{\ast}$ is obtained, which in turn, fixes the miss-detection error probability, and does not provide the flexibility to change the miss-detection rate $\beta$.

\begin{figure}[H]
\centering
\includegraphics[width=3.4 in]{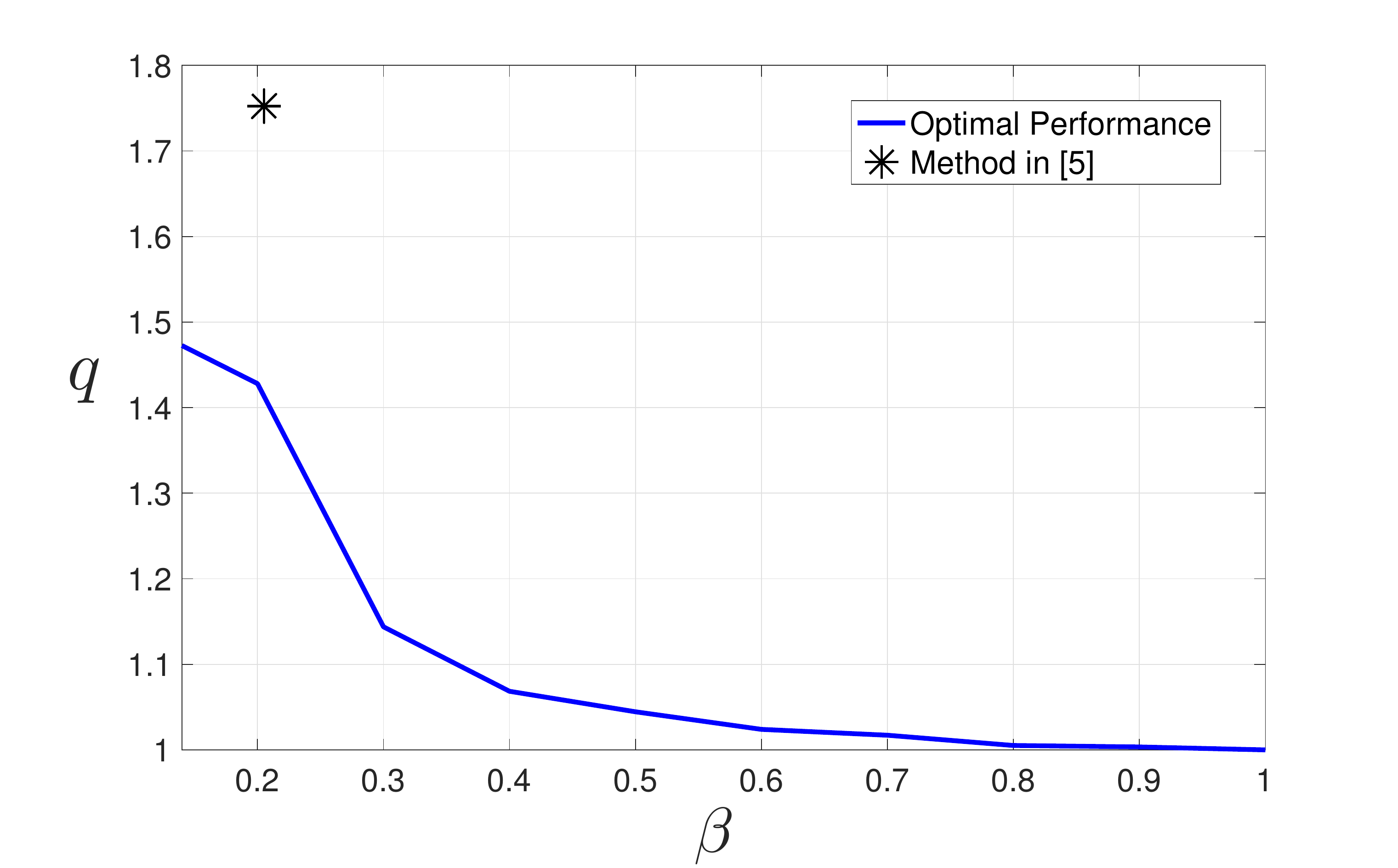}
\caption{$q$ versus $\beta$ for fixed $\alpha^{\ast} = 0.1$.}
\label{fig:1}
\end{figure}

\noindent The results presented in Fig.~\ref{fig:1} correspond to $\sigma = 3$, $\sigma_n = 1$, $h^1 = 1$, $h^2 = 4$, $a = -40, b = 40$. The upper bound on $\Pc_{\sf fa}$ is set to $\alpha^{\ast} = 0.1$, where $\alpha^{\ast}$ is obtained using the methodology in \cite{wilson2016}.

\subsection{Case 2: Both Sensors Vulnerable to Causative Attacks}
We again consider the same model for $X$, and in this setting, we assume that data from both the sensors are vulnerable to being compromised. We assume that the attacker can compromise the data of at most one sensor at any instant. Under such a setting, we have $T=2$, $S_1=\{1\}$, and $S_2=\{2\}$. Therefore, under the adversarial setting, the sensor measurements follow the following composite hypothesis model
\begin{align}\label{eq:hyp:case2}
\begin{array}{ll}
\H_0\ : & \Y \sim f_0(\Y\mid X),\ \text{ with }X \sim \cpi(X)\\
\H_1\ : & \Y \sim f_1(\Y\mid X),\ \text{ with }X \sim \cpi(X)\\
\H_2\ : & \Y \sim f_2(\Y\mid X),\ \text{ with }X \sim \cpi(X)\\ 
\end{array} ,
\end{align}
where $\H_0$ corresponds to the attack-free setting, and hypothesis $\H_i$ corresponds to the data of sensor $i$ being compromised. Motivating by the fact that the sensor with the higher gain $h^i$ is expected to generate a better estimate, we explore a scenario in which the sensor with the higher gain is more likely to be attacked. Hence, we select the parameters $h^1 = 1$, and $h^2 = 2$, and set the probabilities $(\epsilon_0,\epsilon_1,\epsilon_2)=(0.2,0.2,0.6)$. We set the distribution of $X$ to ${\sf Unif}[-2,2]$. We assume that $Y_j^i,$ for $i\in \{1,2\}$, given $X$, is distributed according to ${\cal N}(h^i X,1)$ in the attack-free setting. When sensor $i$ is compromised, we assume that $Y_j^i,$ for $i\in \{1,2\}$, given $X$, follows the distribution ${\cal N}(h^iX,5)$. 

\begin{figure}[H]
\centering
\includegraphics[width=3.4 in]{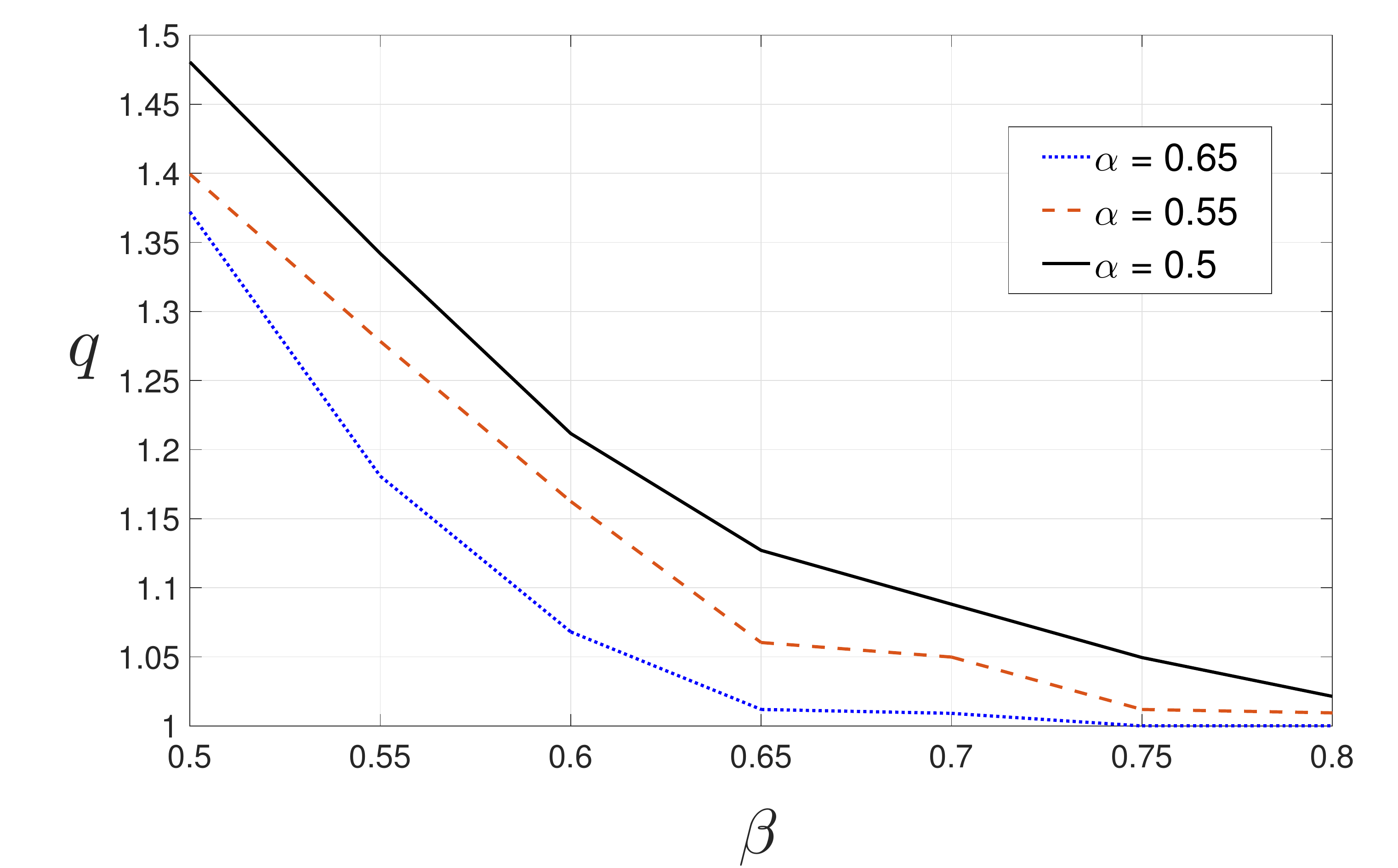}
\caption{$q$ versus $\beta$ for different values of $\alpha$.}
\label{fig:2}
\end{figure}

Figure \ref{fig:2} depicts the performance region defined in Fig.~\ref{fig:PR}, which specifies the variations of $q$ versus $\beta$ for three different values of $\alpha$. The region spanned by the plots between $q$ and $\beta$ for different values of $\alpha$ is the feasible region of operation for the secure estimation methodology developed in this paper. This provides the FC with the flexibility to adjust the emphasis on each of the estimation or detection decisions. As expected, the estimation quality improves monotonically as $\alpha$ and $\beta$ increase.

\subsection{Case 3: Evaluation of Scalable Secure Parameter Estimation Framework}
Before comparing the estimation performance from the scalable approach developed in Section \ref{sec:complexity} and the optimal approach, we illustrate the asymptotic optimality of the decision rule proposed in Theorem \ref{theorem:asympIsolation}. Consider a 2 sensor network, where the measurements of at least one sensor are compromised at any instant. The measurements of the sensors follow a similar model as described in  \eqref{eq:attack_free}. 
Assuming that the parameter $X$ follows the distribution $\N(0,\sigma^2)$, $Y_j^i$ given $X$ is distributed according to $\N(h^i X,\sigma_1^2)$ under the attack-free scenario and according to $\N(h^iX,\sigma_2^2)$ when sensor $i$ is compromised, we illustrate the variations of $-\log(\Pc_{\sf is} (\boldsymbol{ \bar{\delta}}_1))$ versus the number of observations at each sensor in  Fig.~\ref{fig:3}. For the results in Fig.~\ref{fig:3}, we set $h^1 =1, h^2 = 4$, $\sigma^2 = 4$, $\sigma_1^2 = 2$ and $\sigma_2^2 = 6$ . The error probabilities corresponding to both decision rules decay exponentially at the same rate with the increase in number of observations at each sensor. 
\begin{figure}[H]
\centering
\includegraphics[width=3.4 in]{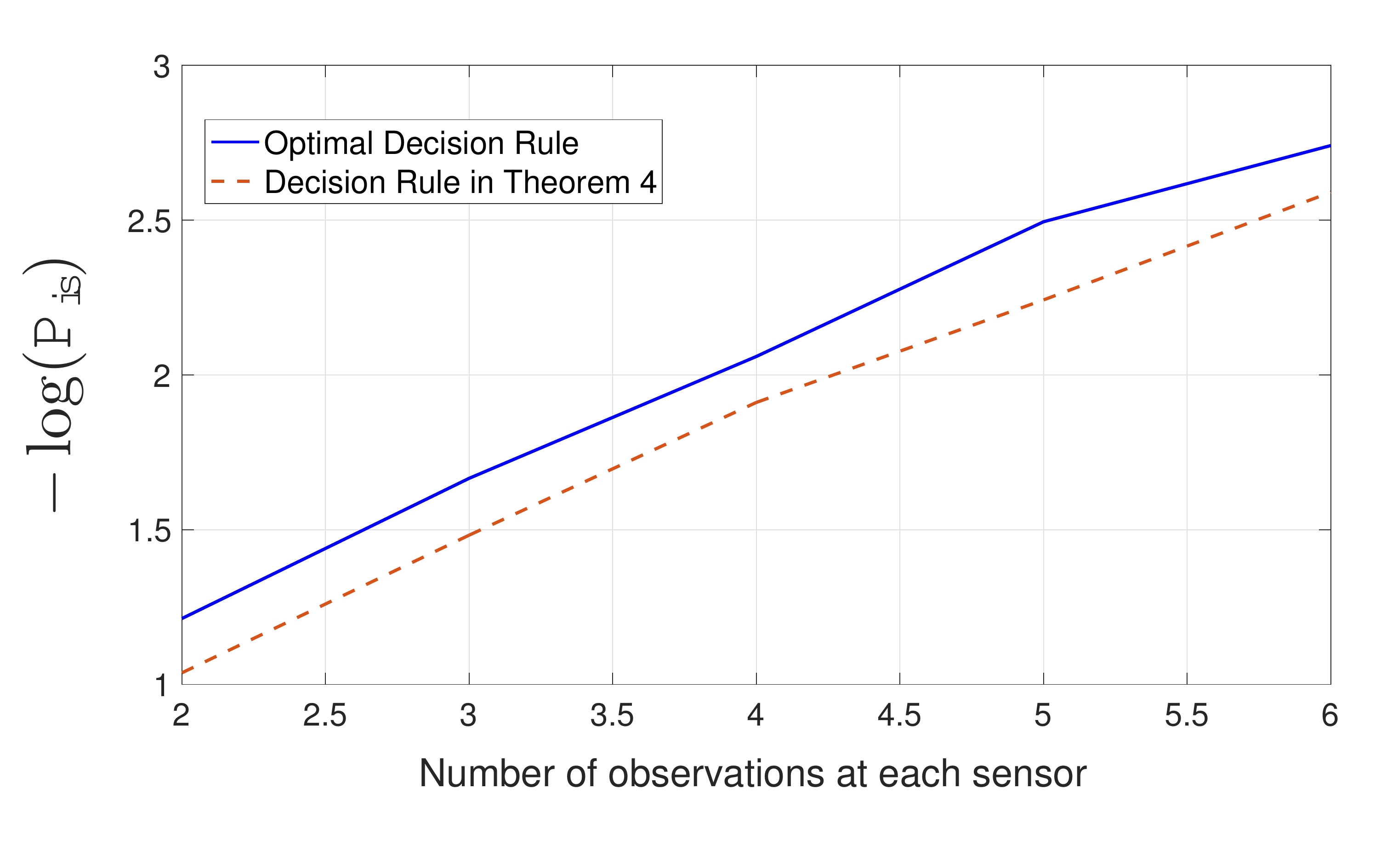}
\caption{$\log(\Pc_{\sf is} (\boldsymbol{ \bar{\delta}}_1))$ versus number of observations at each sensor.}
\label{fig:3}
\end{figure}


In order to compare the performance of the scalable secure estimation framework with that of the optimal framework, we choose $\nu_l^0$ and $\nu_l^1$, for $l \in \{1,2\}$, according to the steps described in Algorithm~2.
We aggregate the reliable local estimates to form an optimal linear estimate at the FC using the fusion strategy for sensor networks with star topology in \cite{fusion}. 
The decision on the reliability of the estimate is formed using the decision rules in Theorem \ref{rel}, following which the local linear estimates from the sensors deemed to provide reliable estimates are aggregated at the FC. To compare with the estimation degradation factor $q$ for an optimal framework, we evaluate the estimation degradation factor for the scalable secure estimation framework and denote it by $\hat{q}$. The estimation performances for the scalable framework and the optimal decision rules are compared in the following table:
\begin{table}[H]
\caption{Comparison of Estimation Performance from Optimal Decision Rules and Heuristic Approach}
\centering
\begin{tabular}{c c c c c c c c}
\hline
$\nu_1^0$ & $\nu_1^1$ & $\nu_2^0$ & $\nu_2^1$ & $\hat{q}$ & $\alpha$ & $\beta$ & $q$\\
\hline
0.2 & 0.57 & 0.2 & 0.5261 & 1.48 & 0.0215 & 0.5134 & 1.126 \\
0.25 & 0.4384 & 0.2 & 0.658 & 1.456 & 0.0216 & 0.6484 & 1.046\\ 
0.3 & 0.4467 & 0.15 & 0.6359 & 1.434  & 0.0482 & 0.5946 & 1.060\\
0.3 & 0.4467 & 0.3 & 0.4124 & 1.491 & 0.0762 & 0.4198 & 1.151\\
0.35 & 0.4021 & 0.25 & 0.4832 & 1.434 & 0.0528 & 0.4812  & 1.136\\
0.15 & 0.689 & 0.35 & 0.4124 & 1.44 & 0.0476 & 0.476 & 1.140
\end{tabular}
\label{table:1}
\end{table}																
\noindent For the results presented in Table \ref{table:1}, we have set $h^1 = h^2 =1$. We assume that the parameter $X$ is distributed according to $\N(0,3)$, $N_j^i$ is distributed according to $\N(0,1)$ under the attack-free scenario, and according to $\N(0,1)\ast{\sf Unif}[-10,10]$ when data from sensor $i$ is compromised. As expected, the optimal decision rules result in superior estimation quality as compared to that obtained from the scalable framework, i.e., $\hat q > q$. 

\section{Conclusion}
We have formalized and analyzed the problem of secure parameter estimation problem under the potential presence of causative attacks on the estimation algorithm. Under causative attacks, the information of the estimation algorithm about the  statistical model of the sampled data is compromised. This leads the estimation algorithm exhibit degraded performance compared to the attack-free setting. We have provided closed-form optimal decision rules that ensure the best estimation quality (minimum estimation cost) while controlling the error in detecting the attacks and isolating the true model of the data. We have shown that the design of optimal estimators is intertwined with the detection rules for deciding upon the true model of the data. Based on this, we have designed the optimal decision rules, which combine both estimation performance and detection power. Based on this vision, the decision-maker can place any desired emphasis on the estimation and detection routines involved. to study the trade-off between the two. We have also provided case studies by applying the theory developed to sensors networks, where sensors face security vulnerabilities. Finally, to circumvent the computational complexity associated with growing the data dimension or attack complexity, we have provided a low-complexity secure estimation algorithm that is optimal in the asymptote of high data dimension.

\appendix
\section{Proof of Theorem \ref{theorem:estimator}}
\label{proof:theorem:estimator}

From (\ref{eqr1}) we have
\begin{align}
\J_i(\delta_i,U_i) &= \E \left[ \C(X,U_i(\Y))\! \mid\! \D\! = \! \H_i\right] \nonumber \\
&= \frac{\displaystyle \int_{\Y} \int_{\X} \delta_i(\Y) \C(X,U_i(\Y))f_i(\Y\!\mid\!X)\cpi(X)\d X \d\Y }{\displaystyle \int_{\Y}\delta_i(\Y)f_i(\Y)\d\Y}\nonumber \;.
\end{align}
Using the definition of $\C_{{\rm p},i}(U_i(\Y)\mid \Y)$ from (\ref{eq:pcost}), a lower bound on $\J_i({\delta_i,U_i(\Y)})$ is given by
\begin{align}\label{eq:e3}
\J_i(\delta_i,U_i) &=\frac{\displaystyle \int_{\Y} \delta_i(\Y)\C_{{\rm p},i}(U_i(\Y)\mid \Y) f_i(\Y) \d\Y }{\displaystyle \int_{\Y}\delta_i(\Y)f_i(\Y)\d\Y}\nonumber \\
& \geq  \frac{\displaystyle \int_{\Y} \delta_i(\Y)\inf_{ U_i(\Y)}\C_{{\rm p},i}(U_i(\Y)\mid \Y) f_i(\Y) \d\Y }{\displaystyle \int_{\Y}\delta_i(\Y)f_i(\Y)\d\Y}\;,
\end{align}
which implies that
\begin{align}
\J_i(\delta_i,U_i)& \geq  \frac{\displaystyle \int_{\Y} \delta_i(\Y){\C}_{{\rm p},i}^{\ast}(\Y)f_i(\Y) \d\Y }{\displaystyle \int_{\Y}\delta_i(\Y)f_i(\Y)\d\Y}\;.
\end{align}
Based on the definition of $\hat{X}_{i}(\Y)$ provided in (\ref{eq:p1}), this lower bound is clearly achieved when the estimator $U_i(\Y)$ is chosen to be
\begin{align}\label{eq:p11}
\hat{X}_i(\Y) = \arg \inf_{U_i(\Y)} \C_{{\rm p},i}(U_i(\Y) \mid\Y)\;,
\end{align}
which proves that the estimator characterized in (\ref{eq:p1}) is an optimal estimator that minimizes the cost $\J_i(\delta_i,U_i)$. The corresponding minimum average estimation cost is
\begin{align}\label{eq:p12}
{\J}_i(\delta_i,\hat{X}_i) = \frac{\displaystyle \int_{\Y} \delta_i(\Y)\C_{{\rm p},i}^{\ast}(\Y)f_i(\Y) \d\Y }{\displaystyle \int_{\Y}\delta_i(\Y)f_i(\Y)\d\Y}\; .
\end{align}
Next, we prove that 
\begin{align}\label{eq:p3}
\max_i \min_{ \bU} \left\lbrace\J_i(\delta_i,U_i)\right\rbrace \equiv \min_{\bU} \max_i \left\lbrace \J_i(\delta_i,U_i)\right\rbrace\; .
\end{align} 
Recall from (\ref{eqr2}), the overall estimation cost $\J(\boldsymbol{\delta},{ \bU}) $ is defined as
\begin{align}
\J(\boldsymbol{\delta},\bU) = \max_i \left\lbrace \J_i(\delta_i,U_i) \right\rbrace\; .
\end{align}
Define $\cC(\boldsymbol{\Omega,\delta},\bU)$ as a convex function of $\J_i(\delta_i,U_i), i \in \left\lbrace0,\dots,T \right\rbrace$, given by
\begin{align}
\cC(\boldsymbol{\Omega},\boldsymbol{\delta},\bU) \triangleq \sum\limits_{i=0}^{T} \Omega_i \J_i(\delta_i,U_i)\; ,
\end{align}
where $\boldsymbol{\Omega} = [\Omega_0,\dots,\Omega_T]$, and $\Omega_i$ satisfy 
\begin{align}
\sum\limits_{i=0}^{T} \Omega_i = 1\; , \text{ and } \Omega_i \in [0,1]\;.
\end{align}
We can represent $\J(\boldsymbol{\delta},\bU)$ as a function of $\cC(\boldsymbol{\Omega},\boldsymbol{\delta},\bU)$ in the following form
\begin{align}
\J(\boldsymbol{\delta},\bU) &= \max_{\boldsymbol{\Omega}} \cC(\boldsymbol{\Omega},\boldsymbol{\delta},\bU)\nonumber \;.
\end{align}
Let $\boldsymbol{\Omega^{\ast}} = \{ \Omega_j^{\ast}:j=0,\dots,T \}$ be defined as
\begin{align}
\boldsymbol{\Omega^{\ast}} &\triangleq \arg \max_{\boldsymbol{\Omega}} \cC(\boldsymbol{\Omega},\boldsymbol{\delta},\bU)\nonumber \;,
\end{align}
where $\Omega_j^{\ast} = 1$ if 
\begin{align}
j = \arg \max_i \left\lbrace \J_i(\delta_i,U_i) \right\rbrace\; .
\end{align}
From (\ref{eq:p11}) and (\ref{eq:p12}), we observe that
\begin{align}\label{eq:cf1}
\max_{\boldsymbol{\Omega}} \min_{\bU}\cC(\boldsymbol{\Omega},\boldsymbol{\delta},\bU) &= \max_{\boldsymbol{\Omega}} \cC(\boldsymbol{\Omega},\boldsymbol{\delta},\hat{\X})\nonumber \\
& \geq \min_{\bU} \max_{\boldsymbol{\Omega}} \cC(\boldsymbol{\Omega},\boldsymbol{\delta},\bU)\;.
\end{align}
Also, at the same time, we have 
\begin{align}
\max_{\boldsymbol{\Omega}} \cC(\boldsymbol{\Omega},\boldsymbol{\delta},\bU) &\geq \max_{\boldsymbol{\Omega}} \min_{\bU}\cC(\boldsymbol{\Omega},\boldsymbol{\delta},\bU) \;,
\end{align}
which implies that
\begin{align}\label{eq:cf2}
\min_{\bU} \max_{\boldsymbol{\Omega}} \cC(\boldsymbol{\Omega},\boldsymbol{\delta},\bU) &\geq \max_{\boldsymbol{\Omega}} \min_{\bU}\cC(\boldsymbol{\Omega},\boldsymbol{\delta},\bU) \;.
\end{align}
From (\ref{eq:cf1}) and (\ref{eq:cf2}), it is easily concluded that
\begin{align}
\max_{\boldsymbol{\Omega}} \min_{\bU}\cC(\boldsymbol{\Omega},\boldsymbol{\delta},\bU) = \min_{\bU} \max_{\boldsymbol{\Omega}} \cC(\boldsymbol{\Omega},\boldsymbol{\delta},\bU) \;,
\end{align}
which completes the proof for (\ref{eq:p3}). Using the results in (\ref{eq:p3}) and (\ref{eq:p12}), the cost function ${\J}(\boldsymbol{\delta}, \hat{\X})$ is given by
\begin{align}
{\J}(\boldsymbol{\delta}, {\hat{\X}}) &= \min_{\bU} \max_{i} \left\lbrace \J_i(\delta_i,U_i) \right\rbrace\nonumber \\
&= \max_i \min_{\bU}\left\lbrace \J_i(\delta_i,U_i) \right\rbrace\nonumber \\
&= \max_i \left\lbrace {\J}_i(\delta_i,\hat{X}_i) \right\rbrace\\
&= \max_i \left\lbrace \frac{\displaystyle \int_{\Y} \delta_i(\Y){\C}_{{\rm p},i}^{\ast}(\Y)f_i(\Y) \d\Y }{\displaystyle \int_{\Y}\delta_i(\Y)f_i(\Y)\d\Y} \right\rbrace\; .
\end{align}

\section{Proof of Theorem~\ref{theorem:convex}}
\label{proof:theorem:convex}
The function $\J_i(\delta_i,U_i)$ is a quasi-convex function. Since the weighted maximum function preserves the quasi-convexity, it is concluded that ${\J}_i(\delta_i,\hat{X}_i) $ is a quasi-convex function from its definition in (\ref{eq:p2}). Therefore, we can find the solution by solving an equivalent feasibility problem given below \cite{boyd2004}. Specifically, for $u \in \mathbb{R_{+}}$, it is easily observed that 
\begin{align}
{\J}(\boldsymbol{\delta}, \hat{\X}) \leq u &\equiv  \int_{\Y} \delta_i(\Y) f_i(\Y) ({\C}^{\ast}_{{\rm p},i}(\Y) - u)\d\Y \leq 0,\text{ for }i\in\left\lbrace0,\dots,T\right\rbrace \;.
\end{align}
Hence, the feasibility problem equivalent to (\ref{eq:prob2}) is given  by
\begin{align}
\Pt(\alpha,\beta) &= \begin{cases}
\min_{\boldsymbol{\delta}} & u\\ \vspace{2pt}
\text{s.t.} & \displaystyle \int_{\Y} \delta_i(\Y) f_i(\Y) ({\C}^{\ast}_{{\rm p},i}(\Y) - u)\d\Y \leq 0, \quad\forall i\in \{ 0,\dots,T\}\\
& \sum\limits_{j=1}^{T}\sum\limits_{i=0,i\neq j}^{T} \frac{\epsilon_j}{1-\epsilon_0}\displaystyle \int_{\Y}\delta_i(\Y)f_j(\Y)\d\Y\  \leq \beta\\
& \sum\limits_{i=1}^{T}\displaystyle \int_{\Y}\delta_i(\Y)f_0(\Y)\d\Y \leq \alpha
\end{cases}\;.
\end{align}
The above problem is feasible if $\Pt(\alpha,\beta) \leq u$, where $\Pt(\alpha,\beta)$ represents the lowest possible value of $u$ for which the problem is feasible and all the constraints are satisfied. Given an interval $[u_0,u_1]$ containing $\Pt(\alpha,\beta)$, the optimal detection rule $\boldsymbol{ \delta}$ and the optimal estimation cost $\Pt(\alpha,\beta)$ can be determined by a bi-section search between $u_0$ and $u_1$ iteratively, solving the feasibility problem in each iteration. 
To solve the feasibility problem, we define an auxiliary convex optimization problem 
\begin{align}\label{eq:ap}
\R(\alpha,\beta,u) &= \begin{cases}
\min_{\boldsymbol{ \delta}} & \eta\\ \vspace{2pt}
\text{s.t.} & \displaystyle \int_{\Y} \delta_i(\Y) f_i(\Y) ({\C}_{{\rm p},i}^{\ast}(\Y) - u)\d\Y \leq \eta, \quad\forall i\in \{ 0,\dots,T\}\\
& \sum\limits_{j=1}^{T}\sum\limits_{i=0,i\neq j}^{T} \frac{\epsilon_j}{1-\epsilon_0}\displaystyle \int_{\Y}\delta_i(\Y)f_j(\Y)\d\Y\  \leq \beta + \eta\\
& \sum\limits_{i=1}^{T}\displaystyle \int_{\Y}\delta_i(\Y)f_0(\Y)\d\Y \leq \alpha + \eta 
\end{cases}\;.
\end{align}
Algorithm 1 summarises the steps for determining $\Pt(\alpha,\beta)$.
\begin{algorithm}
\caption{Bi-section Search}
\label{array-sum}
\begin{algorithmic}[1]
\State \text{Initialize }$u_0, u_1$
\Repeat
    \State $\hat{u} \gets (u_0 + u_1)/2$
    \State \text{Solve }$\R(\alpha,\beta,\hat{u})$
    \If {$\Jc(\alpha,\beta,\hat{u}) \leq 0$}
    \State $u_1 \gets \hat{u}$
    \Else
    \State $u_0 \gets \hat{u}$
    \EndIf
	
\Until $u_1 - u_0 \leq \epsilon$\text{, for }$\epsilon$\text{ sufficiently small}
\State $\Pt(\alpha,\beta)\gets u_1$
\end{algorithmic}
\end{algorithm}
\section{Proof of Theorem \ref{theorem:detector}}\label{proof:theorem:detector}

To solve the problem in (\ref{eq:ap}), a Lagrangian function is constructed according to
\begin{align}
\begin{split}
\Q(\boldsymbol{\delta},\eta,\boldsymbol{\ell}) &\triangleq \left(1-\sum\limits_{i=0}^{T+2} \ell_i\right)\eta\nonumber\\
&\quad + \sum\limits_{i=0}^{T} \ell_i \displaystyle \int_{\Y} \delta_i(\Y)f_i(\Y)({\C}^{\ast}_{{\rm p},i}(\Y) - u)\d\Y \nonumber \\
&\quad + {\ell_{T+1}} \sum\limits_{j=1}^{T}\sum\limits_{i=0,i\neq j}^{T} \frac{\epsilon_j}{1-\epsilon_0}\displaystyle \int_{\Y}\delta_i(\Y)f_j(\Y)\d\Y\ \nonumber- \ell_{T+1}\beta\\
&\quad + \ell_{T+2} \sum\limits_{i=1}^{T}\displaystyle \int_{\Y}\delta_i(\Y)f_0(\Y)\d\Y - \ell_{T+2}\alpha
\end{split}\;,
\end{align}
where $\boldsymbol{\ell} \triangleq \left[ \ell_0,\dots,\ell_{T+2}\right]$ are the non-negative Lagrangian multipliers selected to satisfy the constraints in~\eqref{eq:prob2}, such that 
\begin{align}
\sum\limits_{i=0}^{T+2}\ell_i = 1\;.
\end{align}
The Lagrangian dual function is given by
\begin{align}
d(\boldsymbol{\ell}) &\triangleq \min_{\boldsymbol{ \delta,\eta}} \Q(\boldsymbol{\delta},\eta,\boldsymbol{\ell})\nonumber \\
&= \min_{\boldsymbol{ \delta}} \left(\sum\limits_{i=0}^{T} \displaystyle \int_{\Y} \delta_i(\Y) A_i \d\Y\right) - \ell_{T+1}\beta - \ell_{T+2}\alpha\;,
\end{align}
where 
\begin{align}
A_0 &\triangleq  \ell_0 f_0(\Y) [{\C}^{\ast}_{{\rm p},0}(\Y) - u] + {\ell_{T+1}} \sum\limits_{i=1}^{T} \frac{\epsilon_i}{1-\epsilon_0}f_i(\Y)\;,
\end{align}
and for $i \in \left\lbrace 1,\dots,T \right\rbrace$,
\begin{align}
A_i &\triangleq   \ell_i f_i(\Y) [{\C}^{\ast}_{{\rm p},i}(\Y) - u] + \ell_{T+1}\sum_{j=1,j\neq i}^{T}\frac{\epsilon_j}{1-\epsilon_0} f_j(\Y) +  \ell_{T+2}f_0(\Y)\;.
\end{align}
Therefore, the optimum detection rules that minimize $d(\boldsymbol{\ell})$ are given by:
\begin{align}
\delta_{i}(\Y) = \begin{cases} 
& 1, \quad \text{if }\quad i = i^{\ast}\\
& 0, \quad \text{if }\quad i \neq i^{\ast}
\end{cases},
\end{align}
where
\begin{align}
i^{\ast} = \argmin_{i\in \left\lbrace 0,\dots,T \right\rbrace} A_i  \;.
\end{align}
Hence, the proof is concluded.

\section{Proof of Theorem \ref{lemma:errorprob} }\label{proof:lemma:errorprob}
We can write $\Pc_{\sf is}(\boldsymbol{\hat{\delta}_1})$ as 
\begin{align}
\Pc_{\sf is}(\boldsymbol{\hat{\delta}_1}) =& \mathbb{P}(\D_{\sf is} \neq \T_{\sf is} \mid \D_{\sf d} = \hat\H_1)\nonumber\\
=& \mathbb{P}(\D_{\sf is} \neq \T_{\sf is} \mid \D_{\sf d} = \hat\H_1, \T_{\sf d} = \hat{\H}_0) \mathbb{P}( \T_{\sf d} = \hat{\H}_0 \mid \D_{\sf d} = \hat{\H}_1) \nonumber\\
&\quad + \mathbb{P}(\D_{\sf is} \neq \T_{\sf is} \mid \D_{\sf d} = \hat\H_1, \T_{\sf d} = \hat{\H}_1) \mathbb{P}( \T_{\sf d} = \hat{\H}_1 \mid \D_{\sf d} = \hat{\H}_1)\;.
\end{align}
Note that $ \mathbb{P}(\D_{\sf is} \neq \T_{\sf is} \mid \D_{\sf d} = \hat\H_1, \T_{\sf d} = \hat{\H}_0) = 1$ for any decision rule $\boldsymbol{\hat\delta_1}$, and the terms \newline${\P(\T_{\sf d} = \hat\H_1 \mid \D_{\sf d} = \hat\H_1)}$  and $\P(\T_{\sf d} = \hat\H_0 \mid \D_{\sf d} = \hat\H_1)$ are independent of the decision rule $\boldsymbol{\hat{\delta}}_1$.  Therefore, minimizing $\Pc_{\sf is}(\boldsymbol{\hat{\delta}_1}) $ is equivalent to minimizing $\mathbb{P}(\D_{\sf is} \neq \T_{\sf is} \mid \D_{\sf d} = \hat\H_1, \T_{\sf d} = \hat{\H}_1)$. Furthermore,
\begin{align}
\mathbb{P}(\D_{\sf is} & \neq \T_{\sf is} \mid \D_{\sf d} = \hat\H_1, \T_{\sf d}= \hat{\H}_1) & \nonumber\\
& = \frac{\mathbb{P}(\D_{\sf is} \neq \T_{\sf is} , \D_{\sf d} = \hat\H_1\mid \T_{\sf d} = \hat{\H}_1) \P(\T_{\sf d} = \hat\H_1)}{\mathbb{P}( \D_{\sf d} = \hat\H_1, \T_{\sf d} = \hat{\H}_1)}\\
&= \frac{ \P(\T_{\sf d} = \hat\H_1)}{\mathbb{P}( \D_{\sf d} = \hat\H_1, \T_{\sf d} = \hat{\H}_1)} \nonumber\\
& \qquad \times \sum\limits_{i=1}^T\left( \sum\limits_{{j=1, j\neq i}}^T \P(\D_{\sf is} = \H_i, \D_{\sf d} = \hat\H_1\mid  \T_{\sf is} = \H_j, \T_{\sf d} = \hat\H_1) \P(\T_{\sf is} = \H_j \mid \T_{\sf d} = \hat\H_1)\right)\;\\
&=  \frac{ \P(\T_{\sf d} = \hat\H_1)}{\mathbb{P}( \D_{\sf d} = \hat\H_1, \T_{\sf d} = \hat{\H}_1)}\nonumber\\
& \qquad \times \sum\limits_{i=1}^T\left((1 -   \P(\D_{\sf is} = \H_i, \D_{\sf d} = \hat\H_1\mid  \T_{\sf is} = \H_i, \T_{\sf d} = \hat\H_1)) \P(\T_{\sf is} = \H_i \mid \T_{\sf d} = \hat\H_1)\right)\\
&= \frac{ \P(\T_{\sf d} = \hat\H_1)}{\mathbb{P}( \D_{\sf d} = \hat\H_1, \T_{\sf d} = \hat{\H}_1)} \times \sum\limits_{i=1}^T\left((1 -   \int_{R_i \cap \hat{R}_1} f_i(\Y)\d\Y) \P(\T_{\sf is} = \H_i \mid \T_{\sf d} = \hat\H_1)\right)\;,
\end{align}
where $R_i \subseteq\cY^{n}$ is the region corresponding to the decision $\D_{\sf is} = \H_i$ and $\hat{R}_1\subseteq \cY^{n}$ is the region corresponding to the decision $\D_{\sf d} = \hat\H_1$.  Since we form the decision $\D_{\sf is}$ only when $\D_{\sf d} = \hat{\H}_1$, we conclude that $R_i \subseteq \hat{R}_1$ and hence, $R_i \cap \hat{R}_1 = R_i$. Therefore,
\begin{align}
\mathbb{P}(\D_{\sf is} \neq \T_{\sf is} \mid \D_{\sf d} = \hat\H_1, \T_{\sf d} = \hat{\H}_1) &=  \frac{ \P(\T_{\sf d} = \hat\H_1)}{\mathbb{P}( \D_{\sf d} = \hat\H_1, \T_{\sf d} = \hat{\H}_1)}\nonumber\\ 
&\quad \times\sum\limits_{i=1}^T\left((1 -   \int_{R_i} f_i(\Y)\d\Y) \P(\T_{\sf is} = \H_i \mid \T_{\sf d} = \hat\H_1)\right)\;,
\end{align}
and $\mathbb{P}(\D_{\sf is} \neq \T_{\sf is} \mid \D_{\sf d} = \hat\H_1, \T_{\sf d}) $ is minimized when the regions $R_i$ are chosen using the decision rule
\begin{align}\label{drule1}
\hat \delta_{1i}(\Y) = \begin{cases}
& 1,  \quad \text{\rm if }\;\; i = i^{\ast}\\
& 0, \quad \text{\rm if }\;\; i \neq i^{\ast}
\end{cases}\;,
\end{align}
where
\begin{align}
i^{\ast} = \argmax_{j \in \{1,\dots,T\}} f_j(\Y)\P(\T_{\sf is}= \H_j \mid \T_{\sf d} = \hat\H_1)\;.
\end{align} 

\section{Proof of Theorem \ref{theorem:asympIsolation}}
\label{proof:theorem:isolator}

For the decision rule defined in Theorem \ref{theorem:asympIsolation}, we denote the decision formed as $\bar{\D}_{\sf is} \in \{\H_1,\dots, \H_T\}$  to distinguish it from the decision formed by the optimal decision rule in Lemma \ref{lemma:errorprob}. Also, from Lemma \ref{it},
\begin{align}\label{eqb1}
-\frac{\log {\Pc}_{\sf is}^u}{n} \leq -\frac{\log{\Pc}_{\sf is}(\boldsymbol{\bar\delta_1})}{n} \leq -\frac{\Pc_{\sf is}(\boldsymbol{\hat{\delta}_1})}{n}\leq-\frac{\log{\Pc}_{\sf is}^{ l}}{n} \;.
\end{align} 
Next, we find an upper bound on ${\Pc}_{\sf is}^{ u}$.  When $X$ is guessed to be $X_c$,  the probability that the decision formed by the rule in Theorem \ref{theorem:asympIsolation} is $\H_j$ when the true model is $\H_i$  is given by $\P(\bar\D_{\sf is} = \H_j \mid \T_{\sf is} = \H_i, X_c)$. Assuming that the probabilities $\epsilon_i $, for $i \in \{1,\dots,T\}$, are independent of the choice of $X_c$, we have 
\begin{align}
{\Pc}_{\sf is}^{u} &= \sum\limits_{i = 1}^{T} \sum \limits_{ \substack{i\neq j\\j = 1}}^{T} \epsilon_j \P(\bar{\D}_{\sf is} = \H_i \mid \T_{\sf  is} = \H_j,  X_c)\\
&= \sum\limits_{\substack{j = 2\\ i<j}}^{T} (\epsilon_j + \epsilon_i)\Pc^{ij}\;,
\end{align}
where we have defined
\begin{align}
\Pc^{ij} &\triangleq \frac{\epsilon_j}{\epsilon_i + \epsilon_j}\cdot \P(\bar{\D}_{\sf is} = \H_i \mid \T_{\sf is} = \H_j,  X_c) + \frac{\epsilon_i}{\epsilon_i + \epsilon_j}\cdot \P(\bar{\D}_{\sf is} = \H_j \mid \T_{\sf is} = \H_i, X_c)\;.
\end{align}
Therefore,
\begin{align}\label{ineqf}
{\Pc}_{\sf is}^{u}  &\leq \frac{T(T-1)}{2} \max_{{i\neq j; i,j,\in \{1,\dots T\}}} \Pc^{ij}\;,\\
&\leq \frac{T(T-1)}{2} \max_{{i\neq j; i,j\in \{1,\dots T\}}} \left\lbrace \P(\bar{\D}_{\sf is} = \H_i \mid \T_{\sf is} = \H_j, X_c), \P(\bar{\D}_{\sf is} = \H_j \mid \T_{\sf is} = \H_i, X_c) \right\rbrace\;.
\end{align}
Let $T' \triangleq \frac{T(T-1)}{2}$. Note that
\begin{align}
\P(\bar{\D}_{\sf is} = \H_j \mid \T_{\sf is} = \H_i, X = X_c) \leq \int_{R_{ij}} f_i(\Y\mid X_c)\text{ }\d\Y\;, 
\end{align}
where $R_{ij} = \left\lbrace \Y: \prod_{a \in S_j} \LR_{a} \geq \prod_{b \in S_i} \LR_{b} \right\rbrace$.
Define $B_i$ as the set of coordinates deemed to be compromised in $\H_i$ but not in $\H_j$, i.e., $B_i \triangleq S_i - S_i\cap S_j$, with $\abs{B_i} = r_i$. Similarly, define $B_j \triangleq S_j - S_i \cap S_j$, with $\abs{B_j} = r_j$. Also, since we have $g_l^j$, for $l \in \{1,\dots,n\}, j\in \{0,1\}$, to be conditionally independent distributions given $X_c$, we have 
\begin{align}
f_i(\Y \mid X_c) = \prod\limits_{a \in S_i}\left(\prod\limits_{c=1}^n g_{a}^1(Y_{c}(a)\mid X_c)\right) \prod\limits_{\substack{b \in \bar S_i}}\left(\prod\limits_{d = 1}^n g_{b}^0(Y_{d}(b)\mid X_c)\right)\;,
\end{align}
where $\bar S_i \triangleq \{1,\dots,m\} \backslash S_i$.  Then, the region $R_{ij}$ is equivalent to
\begin{align}
&\left\lbrace \Y: \left( \prod\limits_{a \in B_j} \prod_{c = 1}^n g_{a}^1(Y_{c}(a)\mid X_c)\right) \left( \prod\limits_{b \in B_i} \prod_{d = 1}^n g_{b}^0(Y_{b}(d) \mid X_c)\right) \geq \nonumber \right.\\
&\left.\left( \prod\limits_{a \in B_i} \prod_{c = 1}^n g_{a}^1(Y_{c}(a)\mid X_c)\right)\left( \prod_{b \in B_j}\prod\limits_{d = 1}^n g_{b}^0(Y_{d}(b)\mid X_c)\right) \right\rbrace \;.
\end{align}
Therefore, 
\begin{align}\label{ineq1}
\P(\bar{\D}_{\rm is} = \H_j \mid \T_{\rm is} = \H_i,X_c) &\leq \int_{R_{ij}}\left( \prod\limits_{a \in B_i} \prod_{c = 1}^n g_{a}^1(Y_{c}(a)\mid X_c)\right)\left( \prod_{b \in B_j}\prod\limits_{d = 1}^n g_{b}^0(Y_{d}(b)\mid X_c)\right) \text{ } \d\Y_{B_i\cup B_j}\;,\\
&= \int_{R_{ij}} \min \{t_1,t_2\} \text{ } \d\Y_{B_i \cup B_j}\;,
\end{align}
where $\Y_{B_i \cup B_j}$ is the data for the coordinates in the set $B_i \cup B_j$ and
\begin{align}
t_1&\triangleq \left( \prod\limits_{a \in B_i} \prod_{c = 1}^n g_{a}^1(Y_{c}(a)\mid X_c)\right)\left( \prod_{b \in B_j}\prod\limits_{d = 1}^n g_{b}^0(Y_{d}(b)\mid X_c)\right)\;,\\
t_2&\triangleq \left( \prod\limits_{a \in B_j} \prod_{c = 1}^n g_{a}^1(Y_{c}(a)\mid X_c)\right) \left( \prod\limits_{b \in B_i} \prod_{d = 1}^n g_{b}^0(Y_{b}(d) \mid X_c)\right)\;,
\end{align}
Using the inequality 
\begin{align}
\min(a,b) \leq a^{\lambda}b^{1-\lambda},\qquad\forall\;a,b>0, \quad\text{and }\quad \lambda \in [0,1]\;,
\end{align}
and 
\begin{align}
\int_{R_{ij}}f_i(\Y\mid X_c)\text{ } \d\Y \leq \int f_i(\Y\mid X_c)\text{ } \d\Y\;,
\end{align}
we can modify the inequality in \eqref{ineq1} as
\begin{align}\label{ineq2}
\P(\bar{\D}_{\rm is} = \H_j \mid \T_{\rm is} = \H_i, X_c)  \leq \int t_1^{\lambda}t_2^{1-\lambda}\text{ }\d\Y^{B_i \cup B_j}\;,\qquad\forall\; \lambda \in [0,1]\;.
\end{align}
Following the similar line of analysis as in \eqref{ineq1}-\eqref{ineq2}, it can be verified that
\begin{align}\label{ineq3}
\P(\bar{\D}_{\rm is} = \H_i \mid \T_{\rm is} = \H_j,X_c)\leq  \int t_1^{\lambda}t_2^{1-\lambda}\text{ } \d\Y_{B_i \cup B_j}\;.
\end{align}
The inequalities in \eqref{ineq2} and \eqref{ineq3} hold for all $\lambda \in [0,1]$. Hence, from \eqref{ineqf}, \eqref{ineq2} and \eqref{ineq3}, 
\begin{align}
\Pc_{\sf is}^{ u}&\leq T' \min_{\lambda \in [0,1]} \int t_1^{\lambda} t_2^{1-\lambda} \text{ }\d\Y_{B_i \cup B_j}\;.
\end{align}
{\noindent Under the assumption that the pdfs  of the random variables $Y_a(b) \mid X,$ belong to the location-scale family of distributions with infinite support and only the location parameter a function of $X$}, we use the change of variables $\tilde{Y}_a(b) = Y_a(b) \mid X_c$ to obtain
\begin{align}
\int t_1^{\lambda} t_2^{1-\lambda} \text{ }\d\Y_{B_i \cup B_j} &= \left( \prod\limits_{a \in B_j}  \prod\limits_{c = 1}^n \int (g_{a}^1(\tilde{Y}_{a}(c))) ^{\lambda} (g_a^0(\tilde{Y}_{a}(c)))^{1-\lambda} \text{ }d\tilde{Y}_{a}(c)\right)\nonumber\\
&\quad\times \left( \prod\limits_{b \in B_i} \prod\limits_{d = 1}^k \int (g_b^1(\tilde{Y}_{b}(d)))^{\lambda} (g_{b}^0(\tilde{Y}_{b}(d)))^{1-\lambda} \text{ }d\tilde{Y}_{b}(d)\right)\;.
\end{align}
Note that
\begin{align}
- \frac{\log\Pc_{\sf is}^u}{n} =  - \frac{1}{n}\log\left(\max_{i,j} \min_{\lambda} \int t_1^{\lambda} t_2^{1-\lambda} \text{ } \d\Y_{B_i \cup B_j}\right)\;.
\end{align}
By taking the limit ${n\rightarrow \infty}$ and using the monotonicity of the $\log$ function, we have
\begin{align}
-\lim_{n \rightarrow \infty} \frac{\log\Pc_{\sf is}^u}{n } &= \lim_{n\rightarrow \infty} \frac{1}{n}\min_{i,j}\max_{\lambda} -\log\left(\int t_1^{\lambda} t_2^{1-\lambda} \text{ } \d\Y_{B_i \cup B_j}\right)\\
&= \lim_{n \rightarrow \infty} \frac{1}{n}\min_{i,j}\max_{\lambda} \left( \sum\limits_{a \in B_j}-n\log\left(\int (g_{a}^1(Y))^{\lambda} (g_a^0(Y))^{1-\lambda}\text{ }dY\right)\right. \nonumber\\
&\left.\quad +\sum\limits_{b \in B_i} -n\log\left(\int (g_{b}^0(Y))^{\lambda}(g_{b}^1(Y))^{1-\lambda}\text{ }dY\right)\right)\;.
\end{align}
Therefore,
\begin{align}\label{r}
-\lim_{n \rightarrow \infty} \frac{\log \Pc_{\sf is}^u}{n }& =\min_{i,j} \max_{\lambda} \left( \sum\limits_{a \in B_j}-\log\left(\int g_{a}^1(Y)^{\lambda} g_a^0(Y)^{1-\lambda}\text{ }dY\right) \right. \nonumber\\
& \qquad + \left.\sum\limits_{b \in B_i} -\log\left(\int g_b^0(Y)^{\lambda}g_{b}^1(Y)^{1-\lambda}\text{ }dY\right)\right)\;.
\end{align}
We now find a lower bound on $\Pc_{\sf is}^{\rm l}$. Under the perfect knowledge of $X$, we have
\begin{align}
\Pc_{\sf is}^{ l} &= \sum\limits_{i=1}^T \sum\limits_{\substack{j=1\\j\neq i}}^T \P(\D_{\sf is} = \H_j \mid \T_{\sf is} = \H_i, X) \epsilon_i\;,\\
&= \sum_{i=1}^T \P(\D_{\sf is} \neq \H_i\mid \T_{\sf is} = \H_i,X)\epsilon_i\;,\\
&\geq \max_{i,j} \left\lbrace \P(\D_{\sf is} = \H_j \mid \T_{\sf is} = \H_i, X) \epsilon_i, \P(\D_{\sf is} = \H_i \mid \T_{\sf is} = \H_j, X) \epsilon_j \right\rbrace\;.
\end{align}
Note that
\begin{align}\label{e}
\lim_{n \rightarrow \infty} - \frac{\log(\Pc_{\sf is}^l)}{n} &\leq \lim_{n\rightarrow \infty} -\frac{1}{n}\log\left(\max_{i,j}\left\lbrace  \P(\D_{\sf is} = \H_j \mid \T_{\sf is} = \H_i, X) \epsilon_i, \P(\D_{\sf is} = \H_i \mid \T_{\sf is} = \H_j, X) \epsilon_j\right\rbrace \right)\;,\nonumber \\
&= \lim_{n \rightarrow \infty} - \frac{1}{n} \log\left(\max_{i,j}\left\lbrace  \P(\D_{\sf is} = \H_j \mid \T_{\sf is} = \H_i, X) , \P(\D_{\sf is} = \H_i \mid \T_{\sf is} = \H_j, X) \right\rbrace \right)\;.
\end{align}
For further analysis, we adopt the approach used in \cite{VVV_controlled}. Under model $\H_i$, the data points $Y_w$ are distributed according to the pdf ${f}_i$ , for $w \in \{1,\dots n\}$ and $i \in \{1,\dots T\}$. Define $b_{i,j}^{\lambda}(Y_w) $ as the pdf 
\begin{align}
b_{i,j}^{\lambda}(Y_w) \triangleq \frac{{f}_i^{\lambda} (Y_w \mid X) {f}_j^{1-\lambda}(Y_w \mid X)}{\int {f}_i^{\lambda} (Y_w \mid X) {f}_j^{1-\lambda}(Y_w \mid X) dY_w}\;,
\end{align}
and let 
\begin{align}
\lambda^{\ast} &= \argmax_{\lambda \in [0,1]} -\log \left( \int f_i^{\lambda}(\Y \mid X) f_j^{1-\lambda} (\Y \mid X)\text{ } \d\Y\right)\\
&= \argmax_{\lambda \in [0,1]} -\sum_{w = 1}^n \log \left( \int {f}_i^{\lambda}(Y_w \mid X) {f}_j^{1-\lambda} (Y_w \mid X)\text{ } dY_w\right)\\
&= \argmax_{\lambda \in [0,1]} -n \log \left( \int {f}_i^{\lambda}(Y_w \mid X) {f}_j^{1-\lambda} (Y_w \mid X)\text{ } dY_w\right)\;.
\end{align}
It can be readily verified that
\begin{align}\label{g1}
\max_{\lambda \in [0,1]} -\log \left( \int {f}_i^{\lambda}(Y_w \mid X) {f}_j^{1-\lambda} (Y_w \mid X)\text{ } dY_w\right) 
= D_{\sf KL}(b_{i,j}^{\lambda^{\ast}} \| {f}_i(.\mid X)) =  D_{\sf KL}(b_{i,j}^{\lambda^{\ast}} \| {f}_j(.\mid X))\;,
\end{align}
where $D_{\sf KL}(f\|g)$ denotes the Kullback-Liebler divergence between distributions $f$ and $g$. Define the probability measure $\tilde{\P}$ as
\begin{align}\label{msr}
\tilde{\P}(\Y \mid X) \triangleq \prod\limits_{w = 1}^n b_{i,j}^{\lambda^{\ast}}(Y_w)\;,
\end{align}
and define a random process $M_r$ as 
\begin{align}
M_r \triangleq \sum\limits_{w = 1}^r \left( \log\left( \frac{b_{i,j}^{\lambda^{\ast}}(Y_w)}{{f}_j(Y_w \mid X)}\right) - \E\left[ \log \left(\frac{b_{i,j}^{\lambda^{\ast}}(Y_w)}{{f}_j(Y_w \mid X)} \right) \mid F_{w-1}\right] \right)\;,
\end{align}
where $F_{w-1}$ is the $\sigma-$field generated by $\{Y_1,\dots,Y_{w-1}\}$, and the expectation is with respect to the probability measure $\tilde{\P}$. It can be verified that the process $M_r$ is a stable martingale. According to the martingale stability theorem \cite{loeve}, we have
\begin{align}
\lim_{r\rightarrow \infty} \frac{M_r}{r} \xrightarrow{\text{a.s}} 0\;.
\end{align}
This can be equivalently written in the following form:
\begin{align}\label{e2}
\lim_{r \rightarrow \infty} \tilde{\P}\left( \frac{M_r}{r} > \nu\right) = 0\;, \quad \forall \nu > 0\;.
\end{align}
Since for a given $X$ the random vectors $Y_w$ are i.i.d., we have 
\begin{align}\label{e1}
\E\left[ \log \left(\frac{b_{i,j}^{\lambda^{\ast}}(Y_w)}{{f}_j(Y_w \mid X)} \right) \,\middle\vert\, F_{w-1}\right] &= D_{\sf KL}(b_{i,j}^{\lambda^{\ast}}\| {f}_i(.\mid X)) =  D_{\sf KL}(b_{i,j}^{\lambda^{\ast}}\| {f}_j(.\mid X))\;.
\end{align}
By using \eqref{e1} and restructuring \eqref{e2}, we obtain
\begin{align}\label{rs1}
\lim_{n\rightarrow \infty} \tilde{\P}\left( \prod_{w=1}^n {f}_j(Y_w\mid X) > \exp(-n D_{\sf KL}(b_{i,j}^{\lambda^{\ast}}\| {f}_j(.\mid X)) - n\nu) \times \prod_{w=1}^n b_{i,j}^{\lambda^{\ast}}(Y_w)\right) = 1\;.
\end{align}
Under the probability measure $\tilde{\P}$, we denote the probability of deciding $\H_i$ as the true model when $X$ is given as $\tilde{\P}(\D_{\sf is} = \H_i \mid X)$ and note that either $\tilde{\P}(\D_{\sf is} \neq \H_i\mid X) \geq \frac{1}{2}$ or $\tilde{\P}(\D_{\sf is} = \H_i\mid X) \geq \frac{1}{2}$. Then, using the assumption that $\tilde{\P}(\D_{\sf is} = \H_i\mid X) \geq \frac{1}{2}$ holds, from \eqref{rs1} we get 
\begin{align}\label{e3}
\lim_{n\rightarrow \infty} \tilde{\P}\left( \prod_{w=1}^n {f}_j(Y_w\mid X) > \exp(-n D_{\sf KL}(b_{i,j}^{\lambda^{\ast}}\| {f}_j(.\mid X)) - n\nu) \times \prod_{w=1}^n b_{i,j}^{\lambda^{\ast}}(Y_w), \D_{\sf is} = \H_i\mid X\right) \geq \frac{1}{2} - \kappa\;,
\end{align}
for any $\kappa >0$. Using \eqref{msr}, we conclude that \eqref{e3} is equivalent to
\begin{align}\label{e4}
\lim_{n\rightarrow \infty} \int_R \prod_{w=1}^n b_{i,j}^{\lambda^{\ast}}(Y_w)\text{ } dY_w \geq \frac{1}{2} - \kappa\;, 
\end{align}
where $R$ is the region 
\begin{align}
\left\lbrace \Y: \{\D_{\sf is} = \H_i\mid X\}\text{ and } \left\lbrace \prod_{w=1}^n b_{i,j}^{\lambda^{\ast}}(Y_w\mid X) <\exp(n D_{\sf KL}(b_{i,j}^{\lambda^{\ast}} \| {f}_j(.\mid X)) + n\nu) \times \prod_{w=1}^n {f}_j(Y_w\mid X) \right\rbrace\right\rbrace .
\end{align}
In the region $R$, we have
\begin{align}\label{e5}
\int_R \prod_{w=1}^n b_{i,j}^{\lambda^{\ast}}(Y_w\mid X)\text{ } dY_w &\leq 
\exp(n D_{\sf KL}(b_{i,j}^{\lambda^{\ast}} \| {f}_j(.\mid X)) + n\nu) \int_R \prod_{w=1}^n {f}_j(Y_w\mid X)\d\Y\;,\nonumber\\
&\leq \exp(n D_{\sf KL}(b_{i,j}^{\lambda^{\ast}} \| {f}_j(.\mid X)) + n\nu) \int_{R_2} \prod_{w=1}^n {f}_j(Y_w\mid X)\d\Y\;,\nonumber\\
&= \exp(k D_{\sf KL}(b_{i,j}^{\lambda^{\ast}}\| {f}_j(.\mid X)) + n\nu) \cdot \P(\D_{\sf is} \neq \H_j\mid \T_{\sf is} = \H_j, X)\;,
\end{align}
where region $R_2$ is $\{ \D_{\rm is} \neq \H_j\mid X\}$ and also, $R \subseteq R_2$. From \eqref{e4} and \eqref{e5}, it is concluded that
\begin{align}\label{e6}
\P(\D_{\sf is} \neq \H_j\mid \T_{\sf is} = \H_j,X) \geq \left(\frac{1}{2} - \kappa\right)\exp(-n D_{\sf KL}(b_{i,j}^{\lambda^{\ast}}\| {f}_j(.\mid X)) - n\nu)\;,
\end{align}  
for any $\nu,\kappa >0$. Similarly, if the case $\tilde{\P}(\D_{\rm is} \neq \H_i\mid X) \geq \frac{1}{2}$ holds,
\begin{align}\label{e7}
\P(\D_{\sf is} \neq \H_i\mid \T_{\sf is} = \H_i,X) \geq \left(\frac{1}{2} - \kappa\right)\exp(-n D_{\sf KL}(b_{i,j}^{\lambda^{\ast}} \| {f}_i(.\mid X)) - n\nu)\;.
\end{align}
Using  \eqref{e}, \eqref{e6} and \eqref{e7}, we have
\begin{align}
\lim_{n \rightarrow \infty} - \frac{\log(\Pc_{\sf is}^{ l})}{n} \leq& \lim_{n \rightarrow \infty} -\frac{1}{n} \max_{i,j}\{ \log(\P(\D_{\sf is} \neq \H_j\mid \T_{\sf is} = \H_j,X)), \log(\P(\D_{\sf is} \neq \H_i\mid \T_{\sf is} = \H_i,X))\}\;,\\
=& \lim_{n \rightarrow \infty} \frac{1}{n} \left(-\left(\frac{1}{2} - \kappa\right) + \kappa\nu\right) \nonumber\\
&\quad + \lim_{n \rightarrow \infty} \frac{1}{n} \min_{i,j}\{ n D_{\sf KL}(b_{i,j}^{\lambda^{\ast}}\| {f}_i(.\mid X)), n D_{\sf KL}(b_{i,j}^{\lambda^{\ast}} \| {f}_j(.\mid X))\}\;.
\end{align}
Using the fact that $\kappa$ and $\nu$ can be made arbitrarily close to $0$, and using \eqref{g1}, we get
\begin{align}
\lim_{n \rightarrow \infty} - \frac{\log(\Pc_{\sf is}^{ l})}{n} \leq& \min_{i,j} \max_{\lambda \in [0,1]} -\log \left( \int {f}_i^{\lambda} (Y_w \mid X) {f}_j^{1-\lambda} (Y_w \mid X) \text{ }\d Y_w\right)\\
\nonumber =& \min_{i,j}\max_{\lambda} \left\lbrace \sum_{a \in B_j} -\log\left( \int (g_{a}^1(Y))^{\lambda} (g_a^0(Y))^{1-\lambda} \d Y\right)  \right.  \\
& \qquad\qquad  \quad + \left. \sum_{b \in B_i} -\log\left( \int (g_b^0(Y)^{\lambda} (g_{b}^1(Y))^{1-\lambda}\d Y\right)\right\rbrace\;
 \label{rr}\\
 &= \min_{i\neq j}C(f_i,f_j)\;, \label{expq}
\end{align}
where \eqref{expq} follows from the assumption that the pdfs $g_a^j$  belong to the location-scale family of distributions and have infinite support. From \eqref{r}, \eqref{rr}, it is clear that 
\begin{align}
\lim_{n \rightarrow \infty} - \frac{\log\Pc_{\sf is}^{ l}}{k} = \lim_{n \rightarrow \infty} - \frac{\log{\Pc}_{\sf is}^u}{n}\;.
\end{align} 
Then, using \eqref{eqb1}, we conclude that the error exponents of $\Pc_{\sf is}(\boldsymbol{\hat\delta_1})$ and $\Pc_{\sf is}(\boldsymbol{\bar\delta_1})$ are the same and are given by \eqref{expq}.

\section{Proof of Theorem \ref{rel}}\label{Thm7}
We aim to design the estimator $U_l^j$ and the decision rule $\boldsymbol{\bar\delta}_l$ that minimize the utility function $\J_l^j(\delta_l^j,\bar\delta_l^{\sf r},U_l^j) $ subject to the constraint 
\begin{align}\label{cons}
\P_j(\D_l^{\sf r} = \H_l^{\sf r}) &= \int \delta_l^j(\Y_l) \bar\delta_l^{\sf r}(\Y_l) g_l^j(\Y_l) \d\Y_l \geq 1 - \nu_l^j\;.
\end{align}
Since the effect of estimator only appears in $\Ju$, the optimization problem can be decoupled similarly to the approach followed earlier. The optimal estimator can be readily verified to be 
\begin{align}
\hat{X}^j_l(\Y_l) = \argmin_{\U_l^j} \Ju\;,
\end{align}
where $\hat{X}_{l}^j(\Y_l) $ is defined in (\ref{eq:cest}). In order to design the decision rules, we start by noting that for a decision rule $\boldsymbol{\bar\delta}_l$, such that, 
\begin{align}\label{eq:1}
\int \delta_l^j(\Y_l) \bar\delta_l^{\sf r}(\Y_l) g_l^j(\Y_l) \d\Y_l>1 - \nu_l^j\;,
\end{align}
we can design another decision rule $\boldsymbol{{\Delta}}_l  \triangleq [\Delta_l^{\sf r}(\Y_l), \Delta_l^{\sf u}(\Y_l)]$, such that, $\int \delta_l^j(\Y_l) {\Delta}^{\sf r}_l(\Y_l) g_l^j(\Y_l) \d\Y_l = 1 - \nu_l^j$ with the same estimation performance. To show this, we set 
\begin{align}
{\Delta}^{\sf r}_l(\Y_l) = \frac{(1 - \nu_l^j) \bar\delta^{\sf r}_l(\Y_l)}{\int \delta_l^j(\Y_l) \bar\delta^{\sf r}_l(\Y_l) g_l^j(\Y_l) \d\Y_l }\;,
\end{align}
which satisfies (\ref{cons}) with equality. Using (\ref{eq:1}), note that ${\Delta}^{\sf r}_l(\Y_l) < \bar\delta^{\sf r}_l(\Y_l)$, which implies that $\boldsymbol{{\Delta}}_l$ is a valid decision rule. We can easily verify that
\begin{align}
\J_l^j(\delta_l^j,{\Delta}^{\sf r}_l,\hat{X}_{l}^j) = {{\J_l^j(\delta_l^j,\bar\delta^{\sf r}_l,\hat{X}_l^j)}}\;,
\end{align} 
which implies that the estimation performance is the same for both decision rules, and therefore, we can restrict our design for the optimum decision rule to the class of rules that satisfy (\ref{cons}) with equality. Under the equality condition, $\int \delta_j^i(\Y_i) ,{\bar\delta}^{\sf r}_l(\Y_l) g_l^j(\Y_l) \d\Y_l = 1 - \nu_l^j$, in which case minimizing ${{\J_l^j(\delta_l^j,\bar\delta^{\sf r}_l,\hat{X}_{l}^j)}}$ is equivalent to minimizing $\int \delta_l^j(\Y_l) \bar\delta^{\sf r}_l(\Y_l)\hat{\C}_{l}^j(\Y_l) g_l^j(\Y_l) \d\Y_l$. Let $\gamma_l^j \geq 0$ be the solution to 
\begin{align}
\P_j(\gamma_l^j \geq \hat{\C}_{l}^j(\Y_l)) = \int _{\hat{R}} \delta^j_l(\Y_l) g_l^j(\Y_l) \d\Y_l = 1 - \nu_l^j\;,
\end{align}
where region $\hat{R}$ is $\left\lbrace \Y_l: \gamma_l^j \geq \hat{\C}_{l}^j(\Y_l) \right\rbrace$. Then,
\begin{align}
&\int \delta^j_l(\Y_l) \bar\delta^{\sf r}_l(\Y_l)\hat{\C}_{l}^j(\Y_l) g_l^j(\Y_l) \d\Y_l - \gamma_l^j(1-\nu_l^j) \nonumber\\
&= \int \delta^j_l(Y_l) \bar\delta^{\sf r}_l(\Y_l)\hat{\C}_{l}^j(\Y_l) g_l^j(\Y_l) \d\Y_l - \gamma_l^j\int \delta^j_l(\Y_l) \bar\delta^{\sf r}_l(\Y_l)g_l^j(\Y_l) \d\Y_l\nonumber \\
&= \int \delta^j_l(\Y_l) \bar\delta^{\sf r}_l(\Y_l)(\hat{\C}_{l}^j(\Y_l)-\gamma_l^j) g_l^j(\Y_l) \d\Y_l\nonumber \\
& \geq \int_{\hat{R}} \delta^l_j(\Y_l) (\hat{\C}_{l}^j(\Y_l)-\gamma_l^j) g_l^j(\Y_l) \d\Y_l \nonumber\\
& = \int_{\hat{R}} \delta^l_j(\Y_l) \hat{\C}_{l}^j(\Y_l) g_l^j(\Y_l) \d\Y_l -\gamma_l^j \P_j(\gamma_l^j \geq \hat{\C}_{l}^j(\Y_l)) \nonumber\\
&= \int_{\hat{R}} \delta^j_l(\Y_l) \hat{\C}_{l}^j(\Y_l) g_l^j(\Y_l) \d\Y_l -\gamma_l^j (1-\nu_l^j)\;.
\end{align}
Clearly, $\int \delta^j_l(\Y_l) \bar\delta^{\sf r}_l(\Y_l)\hat{\C}_{l}^j(\Y_l) g_l^j(\Y_l) \d\Y_l  \geq  \int_{\hat{R}} \delta^j_l(\Y_l) \hat{\C}_{l}^j(\Y_l) g_l^j(\Y_l) \d\Y_l $ and therefore, the decision rule $\bar\delta^{\sf r}_l(\Y_l)$ given by
\begin{align}
\bar\delta_l^{\sf r}(\Y_l) = \mathbbm{1}_{\left\lbrace \gamma_l^j \geq \hat{\C}_{l}^j(\Y_l) \right\rbrace}\;,
\end{align}
is optimal as it ensures optimal estimation performance and satisfies the constraint in (\ref{cons}).

\bibliographystyle{IEEEtran}
\bibliography{Secure_Estimation}

\begin{thebibliography}{10}
\providecommand{\url}[1]{#1}
\csname url@samestyle\endcsname
\providecommand{\newblock}{\relax}
\providecommand{\bibinfo}[2]{#2}
\providecommand{\BIBentrySTDinterwordspacing}{\spaceskip=0pt\relax}
\providecommand{\BIBentryALTinterwordstretchfactor}{4}
\providecommand{\BIBentryALTinterwordspacing}{\spaceskip=\fontdimen2\font plus
\BIBentryALTinterwordstretchfactor\fontdimen3\font minus
  \fontdimen4\font\relax}
\providecommand{\BIBforeignlanguage}[2]{{%
\expandafter\ifx\csname l@#1\endcsname\relax
\typeout{** WARNING: IEEEtran.bst: No hyphenation pattern has been}%
\typeout{** loaded for the language `#1'. Using the pattern for}%
\typeout{** the default language instead.}%
\else
\language=\csname l@#1\endcsname
\fi
#2}}
\providecommand{\BIBdecl}{\relax}
\BIBdecl

\bibitem{J21}
A.~Tajer, V.~V. Veeravalli, and H.~V. Poor, ``{Outlying Sequence Detection in
  Large Datasets - {A} Data-Drive Approach},'' \emph{IEEE Signal Processing
  Magazine - Special Issue on Signal Processing for Big Data}, vol.~31, no.~5,
  pp. 44--56, September 2014.

\bibitem{vempaty}
A.~Vempaty, L.~Tong, and P.~K. Varshney, ``Distributed inference with
  {B}yzantine data: {S}tate-of-the-art review on data falsification attacks,''
  \emph{IEEE Signal Processing Magazine}, vol.~30, no.~5, pp. 65--75, Sep.
  2013.

\bibitem{vempatytarget}
A.~Vempaty, Y.~S. Han, and P.~K. Varshney, ``Target localization in wireless
  sensor networks using error correcting codes,'' \emph{IEEE Transactions on
  Information Theory}, vol.~60, no.~1, pp. 697--712, Jan. 2014.

\bibitem{Barreno}
M.~Barreno, B.~Nelson, R.~Sears, A.~D. Joseph, and J.~D. Tygar, ``Can machine
  learning be secure?'' in \emph{Proc. ACM Symposium on Information, computer
  and communications security}, Taipei, Taiwan, Mar. 2006, pp. 16--25.

\bibitem{wilson2016}
C.~Wilson and V.~V. Veeravalli, ``{MMSE} estimation in a sensor network in the
  presence of an adversary,'' in \emph{Proc. IEEE International Symposium on
  Information Theory}, Barcelona, Spain, Jul. 2016, pp. 2479--2483.

\bibitem{vempatylocalization}
A.~Vempaty, O.~Ozdemir, K.~Agrawal, H.~Chen, and P.~K. Varshney, ``Localization
  in wireless sensor networks: {B}yzantines and mitigation techniques,''
  \emph{IEEE Transactions on Signal Processing}, vol.~61, no.~6, pp.
  1495--1508, Mar. 2013.

\bibitem{manet}
P.~Ebinger and S.~D. Wolthusen, ``Efficient state estimation and {B}yzantine
  behavior identification in tactical {MANET}s,'' in \emph{Proc. IEEE Military
  Communications Conference}, Boston, MA, Oct. 2009.

\bibitem{asymptote}
J.~Zhang, R.~S. Blum, X.~Lu, and D.~Conus, ``Asymptotically optimum distributed
  estimation in the presence of attacks,'' \emph{IEEE Transactions on Signal
  Processing}, vol.~63, no.~5, pp. 1086--1101, Mar. 2015.

\bibitem{distributedlarge}
J.~Zhang and R.~S. Blum, ``Distributed estimation in the presence of attacks
  for large scale sensor networks,'' in \emph{Proc. Conference on Information
  Sciences and Systems}, Princeton, NJ, Mar. 2014.

\bibitem{rawat}
A.~S. Rawat, P.~Anand, H.~Chen, and P.~K. Varshney, ``Countering {B}yzantine
  attacks in cognitive radio networks,'' in \emph{Proc. IEEE International
  Conference on Acoustics, Speech and Signal Processing}, Dallas, TX, Mar.
  2010, pp. 3098--3101.

\bibitem{decentralized}
E.~Soltanmohammadi, M.~Orooji, and M.~Naraghi-Pour, ``Decentralized hypothesis
  testing in wireless sensor networks in the presence of misbehaving nodes,''
  \emph{IEEE Transactions on Information Forensics and Security}, vol.~8,
  no.~1, pp. 205--215, Jan. 2013.

\bibitem{vempatyadaptive}
A.~Vempaty, K.~Agrawal, P.~Varshney, and H.~Chen, ``{Adaptive learning of
  Byzantines' behavior in cooperative spectrum sensing},'' in \emph{Proc. IEEE
  Wireless Communications and Networking Conference}, Cancun, Mexico, Mar.
  2011, pp. 1310--1315.

\bibitem{cps1}
H.~Fawzi, P.~Tabuada, and S.~Diggavi, ``Secure state-estimation for dynamical
  systems under active adversaries,'' in \emph{Proc. Allerton Conference on
  Communication, Control, and Computing}, Monticello, IL, Sep. 2011, pp.
  337--344.

\bibitem{cps2}
------, ``Secure estimation and control for cyber-physical systems under
  adversarial attacks,'' \emph{IEEE Transactions on Automatic Control},
  vol.~59, no.~6, pp. 1454--1467, Jun. 2014.

\bibitem{cps3}
S.~Z. Yong, M.~Zhu, and E.~Frazzoli, ``Resilient state estimation against
  switching attacks on stochastic cyber-physical systems,'' in \emph{Proc. IEEE
  Conference on Decision and Control}, Osaka, Japan, Dec. 2015, pp. 5162--5169.

\bibitem{cps4}
M.~Pajic, P.~Tabuada, I.~Lee, and G.~J. Pappas, ``Attack-resilient state
  estimation in the presence of noise,'' in \emph{Proc. IEEE Conference on
  Decision and Control}, Osaka, Japan, Dec. 2015, pp. 5827--5832.

\bibitem{cps5}
Y.~Shoukry, P.~Nuzzo, A.~Puggelli, A.~L. Sangiovanni-Vincentelli, S.~A. Seshia,
  and P.~Tabuada, ``Secure state estimation for cyber physical systems under
  sensor attacks: a satisfiability modulo theory approach,'' \emph{IEEE
  Transactions on Automatic Control}, vol.~62, no.~10, pp. 4917--4932, Mar.
  2017.

\bibitem{2015secure}
S.~Mishra, Y.~Shoukry, N.~Karamchandani, S.~Diggavi, and P.~Tabuada, ``Secure
  state estimation: {O}ptimal guarantees against sensor attacks in the presence
  of noise,'' in \emph{Proc. IEEE International Symposium on Information
  Theory}, Hong Kong, China, Jun. 2015, pp. 2929--2933.

\bibitem{cps6}
M.~Pajic, J.~Weimer, N.~Bezzo, P.~Tabuada, O.~Sokolsky, I.~Lee, and G.~J.
  Pappas, ``Robustness of attack-resilient state estimators,'' in \emph{Proc.
  IEEE International Conference on Cyber-Physical Systems}, Apr. 2014, pp.
  163--174.

\bibitem{kalman}
C.~Z. Bai and V.~Gupta, ``On {K}alman filtering in the presence of a
  compromised sensor: {F}undamental performance bounds,'' in \emph{Proc.
  American Control Conference}, Portland, OR, Jun. 2014, pp. 3029--3034.

\bibitem{middleton}
D.~Middleton and R.~Esposito, ``Simultaneous optimum detection and estimation
  of signals in noise,'' \emph{IEEE Transactions on Information Theory},
  vol.~14, no.~3, pp. 434--444, May 1968.

\bibitem{1992}
O.~Zeitouni, J.~Ziv, and N.~Merhav, ``When is the generalized likelihood ratio
  test optimal?'' \emph{IEEE Transactions on Information Theory}, vol.~38,
  no.~5, pp. 1597--1602, Sep. 1992.

\bibitem{2012joint}
G.~V. Moustakides, G.~H. Jajamovich, A.~Tajer, and X.~Wang, ``Joint detection
  and estimation: Optimum tests and applications,'' \emph{IEEE Transactions on
  Information Theory}, vol.~58, no.~7, pp. 4215--4229, Jul. 2012.

\bibitem{2012minimax}
G.~H. Jajamovich, A.~Tajer, and X.~Wang, ``Minimax-optimal hypothesis testing
  with estimation-dependent costs,'' \emph{IEEE Transactions on Signal
  Processing}, vol.~60, no.~12, pp. 6151--6165, Dec. 2012.

\bibitem{poor2013}
H.~V. Poor, \emph{An {I}ntroduction to {S}ignal {D}etection and {E}stimation},
  2nd~ed.\hskip 1em plus 0.5em minus 0.4em\relax New York: Springer-Verlag,
  1998.

\bibitem{distest3}
A.~S. Behbahani, A.~M. Eltawil, and H.~Jafarkhani., ``Decentralized estimation
  under correlated noise.'' \emph{{IEEE Transactions on Signal Processing}},
  vol.~62, no.~21, pp. 5603--5614, 2014.

\bibitem{lineardec}
A.~S. Behbahani, A.~M. Eltawil, and J.~Hamid, ``{Linear decentralized
  estimation of correlated data for power-constrained wireless sensor
  networks},'' \emph{{IEEE Transactions on Signal Processing}}, vol.~60,
  no.~11, pp. 6003--6016, 2012.

\bibitem{fusion}
\BIBentryALTinterwordspacing
X.~Yuzhe, V.~Gupta, and C.~Fischione, ``{D}istributed estimation,'' Tech. Rep.,
  2012. [Online]. Available:
  \url{https://www3.nd.edu/~vgupta2/research/publications/xgf13.pdf}
\BIBentrySTDinterwordspacing

\bibitem{xiao2005}
X.~Jin-Jun and L.~Zhi-Quan, ``{Decentralized estimation in an inhomogeneous
  sensing environment},'' \emph{{IEEE Transactions on Information Theory}},
  vol.~51, no.~10, pp. 3564--3575, 2005.

\bibitem{boyd2004}
S.~Boyd and L.~Vandenberghe, \emph{Convex {O}ptimization}.\hskip 1em plus 0.5em
  minus 0.4em\relax Cambridge University Press, 2004.

\bibitem{VVV_controlled}
N.~Sirin, G.~K. Atia, and V.~V. Veeravalli, ``Controlled {S}ensing for
  {M}ultihypothesis {T}esting.'' \emph{IEEE Transactions on Automatic Control},
  vol.~58, no.~10, pp. 2451--2464, 2013.

\bibitem{loeve}
M.~Loeve, \emph{Probability {T}heory {II}}.\hskip 1em plus 0.5em minus
  0.4em\relax Springer, 1978.

\end{thebibliography}

\end{document}